\documentclass[a4paper,UKenglish,cleveref, autoref, thm-restate]{lipics-v2021}
\hideLIPIcs  %uncomment to remove references to LIPIcs series (logo, DOI, ...), e.g. when preparing a pre-final version to be uploaded to arXiv or another public repository

\title{On Kernels for \textsc{$d$-Path Vertex Cover}}

\author{Radovan Červený}{Department of Theoretical Computer Science, Faculty of Information Technology,
\\Czech Technical University in Prague, Prague, Czech~Republic}{radovan.cerveny@fit.cvut.cz}{https://orcid.org/0000-0003-4528-9525
}{The author acknowledges the support of the Grant Agency of the Czech
Technical University in Prague, grant No. SGS20/208/OHK3/3T/18.}%
\author{Pratibha Choudhary}{Department of Theoretical Computer Science, Faculty of Information Technology,
\\Czech Technical University in Prague, Prague, Czech~Republic}{ondrej.suchy@fit.cvut.cz}{https://orcid.org/0000-0002-1648-288X}{The author acknowledges the support of the OP VVV MEYS funded project
CZ.02.1.01/0.0/0.0/16\_019/0000765 ``Research Center for Informatics''.}

\author{Ond\v rej~Such\'{y}}{Department of Theoretical Computer Science, Faculty of Information Technology,
\\Czech Technical University in Prague, Prague, Czech~Republic}{ondrej.suchy@fit.cvut.cz}{https://orcid.org/0000-0002-7236-8336}{The author acknowledges the support of the OP VVV MEYS funded project
CZ.02.1.01/0.0/0.0/16\_019/0000765 ``Research Center for Informatics''.}

\authorrunning{R. Červený and P. Choudhary and O. Such\'{y}}
\Copyright{Radovan Červený and Pratibha Choudhary and Ond\v rej~Such\'{y}} %TODO mandatory, please use full first names. LIPIcs license is "CC-BY";  http://creativecommons.org/licenses/by/3.0/

 %TODO mandatory: Please choose ACM 2012 classifications from https://dl.acm.org/ccs/ccs_flat.cfm

\category{} %optional, e.g. invited paper

\nolinenumbers %uncomment to disable line numbering

\EventEditors{Stefan Szeider, Robert Ganian, and Alexandra Silva}
\EventNoEds{3}
\EventLongTitle{47th International Symposium on Mathematical Foundations of Computer Science (MFCS 2022)}
\EventShortTitle{MFCS 2022}
\EventAcronym{MFCS}
\EventYear{2022}
\EventDate{August 22--26, 2022}
\EventLocation{Vienna, Austria}
\EventLogo{}
\SeriesVolume{241}
\ArticleNo{81}
\newcommand*{\ISLIPICS}{}%

\usepackage{algorithmic}
\usepackage[ruled,vlined,boxed,linesnumbered]{algorithm2e}
\usepackage{amsmath,amssymb}
\usepackage{etoolbox}

\usepackage{hyperref}
\usepackage{nameref}
\usepackage{tabularx}
\usepackage{enumerate}
\usepackage{lineno}

\makeatletter
\let\original@algocf@latexcaption\algocf@latexcaption
\long\def\algocf@latexcaption#1[#2]{%
  \@ifundefined{NR@gettitle}{%
    \def\@currentlabelname{#2}%
  }{%
    \NR@gettitle{#2}%
  }%
  \original@algocf@latexcaption{#1}[{#2}]%
}
\makeatother

\usepackage{xspace}
\usepackage{tikz}
\usetikzlibrary{calc,fit,backgrounds,decorations.pathmorphing,decorations.pathreplacing,positioning,arrows.meta,snakes}

\newcommand{\ostar}{\mathcal{O}^*}
\newcommand{\bigo}{\mathcal{O}}
\newcommand{\PP}{\mathcal{P}}
\newcommand{\PVC}{\textsc{PVC}}
\newcommand{\dPVC}{$d$-\textsc{PVC}}

\ifdefined\ISLIPICS
\else
    \newtheorem{observation}{Observation}
    
\fi

\newtheorem{rrule}{Reduction Rule}

\ifdefined\ISLIPICS
\else
    
\fi

\ifdefined\ISLIPICS
\else
    \newenvironment{claimproof}{{\emph{Proof of the claim.} }}{\hfill$\vartriangleleft$}
\fi

\ifdefined\ISLIPICS
\else
    \AtEndEnvironment{proof}{\phantom{}\qed}
\fi

\newcommand{\M}{\mathcal{M}}
\newcommand{\F}{\mathcal{F}}

\newcommand{\R}{\mathcal{R}}
\newcommand{\Y}{\mathcal{Y}}
\newcommand{\anc}{\textit{anc}}
\newcommand{\sub}{\textit{sub}}
\DeclareRobustCommand*{\req}[1]{\varrho#1}
\DeclareRobustCommand*{\subreq}[1]{\sigma#1}

\definecolor{vertexVorangeEdge}{RGB}{245, 182, 66}
\definecolor{vertexVorangeEdgeDim}{RGB}{255, 236, 201}

\definecolor{matchingMBoxDraw}{RGB}{0,0,0}
\definecolor{matchingMBoxFill}{RGB}{234, 234, 234}
\definecolor{matchingMvertex}{RGB}{0,0,0}

\definecolor{setXBoxDraw}{RGB}{0, 255, 238}
\definecolor{setXBoxFill}{RGB}{235, 255, 254}
\definecolor{setXvertex}{RGB}{0, 255, 238}

\definecolor{setX2edge}{RGB}{252, 56, 255}
\definecolor{setX2BoxDraw}{RGB}{252, 56, 255}
\definecolor{setX2BoxFill}{RGB}{255, 237, 255}

\definecolor{setX0BoxDraw}{RGB}{87, 201, 83}
\definecolor{setX0BoxFill}{RGB}{216, 255, 214}

\definecolor{setX1edge}{RGB}{61, 71, 255}
\definecolor{setX1BoxDraw}{RGB}{61, 71, 255}
\definecolor{setX1BoxFill}{RGB}{227, 228, 255}

\definecolor{setM1BoxDraw}{RGB}{162, 219, 77}
\definecolor{setM1BoxFill}{RGB}{239, 255, 214}

\definecolor{setMprimeBoxDraw}{RGB}{54, 156, 0}
\definecolor{setMprimeBoxFill}{RGB}{162, 219, 77}

\definecolor{setXprimeBoxDraw}{RGB}{0, 24, 130}
\definecolor{setXprimeBoxFill}{RGB}{140, 156, 255}
\definecolor{setXprimeEdge}{RGB}{252, 56, 255}

\definecolor{NFiBoxDraw}{RGB}{245, 182, 66}
\definecolor{NFiBoxFill}{RGB}{255, 246, 221}

\definecolor{CiBoxDraw}{RGB}{255, 125, 134}
\definecolor{CiBoxFill}{RGB}{255, 240, 241}

\usepackage{etoolbox}

\newcommand{\lv}[1]{}

\newcommand{\appendixText}{}

\newcommand{\toappendix}[1]{\gappto{\appendixText}{{#1}}}

\newcommand{\appmark}{$\star$}

\newcommand{\ifappendix}[1]{#1}

\newcommand{\mainappendix}{}

\begin{document}

\maketitle

\begin{abstract}
In this paper we study the kernelization of the $d$-\textsc{Path Vertex Cover} ($d$-PVC) problem. Given a graph $G$, the problem requires finding whether there exists a set of at most $k$ vertices whose removal from $G$ results in a graph that does not contain a path (not necessarily induced) with $d$ vertices. It is known that $d$-PVC is \textsc{NP}-complete for $d\geq 2$. Since the problem generalizes to $d$-\textsc{Hitting Set}, it is known to admit a kernel with $\mathcal{O}(dk^d)$ edges. We improve on this by giving better kernels.
Specifically, we give kernels with $\mathcal{O}(k^2)$ vertices and edges for the cases when $d=4$ and $d=5$.
Further, we give a kernel with $\mathcal{O}(k^4d^{2d+9})$ vertices and edges for general $d$.

\end{abstract}

\section{Introduction}\label{sec:introduction}
Vertex deletion problems have been studied extensively in graph theory. These problems require finding a subset of vertices whose deletion results in a graph that belongs to some desired class of graphs. One such problem is path covering.
Given a graph $G=(V,E)$, the \textsc{$d$-Path Vertex Cover} problem ($d$-PVC) asks to compute a subset $S\subseteq V$ of vertices such that the graph resulting from removal of $S$ does not contain a path on $d$ vertices. Here the path need not necessarily be induced. The problem was first introduced by Brešar et al.~\cite{BresarKKS11}. It is known to be NP-complete for any $d \ge 2$ due to the meta-theorem of Lewis and Yannakakis~\cite{LewisY80}. The \textsc{2-PVC} problem is the same as the well known \textsc{Vertex Cover} problem. The \textsc{3-PVC} problem is also known as \textsc{Maximum Dissociation Set} or \textsc{Bounded Degree-One Deletion}. The \textsc{$d$-PVC} problem is motivated by the field of designing secure wireless communication protocols~\cite{Novotny10} or in route planning and speeding up of shortest path queries~\cite{FunkeNS16}.

With respect to exact algorithms, several efficient (better than brute force enumeration) algorithms are known for \textsc{2-PVC} and \textsc{3-PVC}. In particular, \textsc{2-PVC} (\textsc{Vertex Cover}) can be solved in $\bigo(1.1996^n)$ time and polynomial space due to Xiao and Nagamochi~\cite{XiaoN17}, while \textsc{3-PVC} can be solved in $\bigo(1.4613^n)$ time and polynomial space due to Chang et al.~\cite{ChangCHLRS18} or in $\bigo(1.3659^n)$ time and exponential space due to Xiao and Kou~\cite{XiaoK17exact}.

From the approximation point of view, it is known due to Brešar et al.~\cite{BresarKKS11} that \textsc{$d$-PVC}, for $d > 2$, cannot be $r$-approximated within a factor of $r=1.3606$ in polynomial time, unless P$=$NP.
A greedy $d$-approximation algorithm for \textsc{$d$-PVC} can be employed by repeatedly finding a $d$-path and putting its vertices into the solution.
Due to Fomin et al.~\cite{rep-sets}, we can find an arbitrary $d$-path in $\bigo(2.619^ddn\log{n})$ time, and therefore the approximation algorithm runs in $\bigo(n^2\log{n})$ time in the size of the input.
While the algorithms of Zehavi~\cite{Zehavi15KPath} and Tsur~\cite{Tsur19bKPath},
with running times $\ostar({2.597^d})$ and $\ostar({2.554^d})$,\footnote{The $\ostar()$ notation suppresses all factors polynomial in the input size.} respectively, can be faster for large~$d$,
their running time factor polynomial in input size is much worse than $\bigo(n\log{n})$.
Lee~\cite{Lee19} gave a $\bigo(\log{d})$-approximation  algorithm which runs in $\ostar(2^{\bigo(d^3\log{d})})$ time.
For \textsc{3-PVC} a $2$-approximation algorithm was given by Tu and Zhou~\cite{TuZ11} and for \textsc{4-PVC} a $3$-approximation algorithm is known due to Camby et al.~\cite{Camby14}.

When parameterized by the size of the solution~$k$, \textsc{$d$-PVC} is directly solvable by a~trivial FPT algorithm for $d$-\textsc{Hitting Set}, that runs in $\ostar(d^k)$ time.
However, since \textsc{$d$-PVC} is a restricted case of \textsc{$d$-Hitting Set}, it is known due to Fomin et al.~\cite{FominGKLS10} that for  $d\ge4$ \textsc{$d$-PVC} can be solved in $\ostar((d - 0.9245)^{k})$ time
and for $d \ge 6$ algorithms with even better running times are known due to Fernau~\cite{Fernau10}.
Namely the running times are $\ostar((d-1+c_d)^{k})$, where $c_d$ is a small positive constant which monotonically approaches $0$ as $d$ goes to $\infty$.
There has been considerable study for the case when $d$ is a small constant. For the \textsc{2-PVC} (\textsc{Vertex Cover}) problem, current best known algorithm due to Chen, Kanj, and Xia~\cite{ChenKX10} runs in time $\ostar(1.2738^k)$. For \textsc{3-PVC}, %Tu~\cite{Tu15} derived an iterative compression technique to solve the problem in $\ostar(2^k)$ time.
%This was later improved by Katrenič~\cite{Katrenic16} to $\ostar(1.8127^k)$, by Xiao and Kou~\cite{XiaoK17} to $\ostar(1.7485^k)$ by using a~branch-and-reduce approach and finally by
the current best known algorithm due to Tsur~\cite{Tsur19a} runs in $\ostar(1.713^k)$ time. For the \textsc{4-PVC} problem, Tsur~\cite{Tsur21-4PVC} gave the current best algorithm that runs in $\ostar(2.619^k)$ time.
In previous work~\cite{CervenyS19,CervenyS22-5PVC}, a subset of authors developed an $\ostar(4^k)$ algorithm for \textsc{$5$-PVC}. For $d=5$, $6$, and~$7$ Tsur~\cite{Tsur19b} claimed algorithms for $d$-PVC with running times $\ostar(3.945^k)$, $\ostar(4.947^k)$, and $\ostar(5.951^k)$, respectively.
A subset of authors used a computer to generate even faster algorithms for $3\le d\le 8$~\cite{CervenyS21-generating}.

In this paper, we are interested in kernels for the {\dPVC} problem. Since an instance of {\dPVC} can be formulated as an instance of \textsc{$d$-Hitting Set}, by using the results of Fafianie and Kratsch~\cite{FafianieK15} we immediately get a kernel for {\dPVC} with at most $d(k+1)^{d}$ vertices and at most $(d-1)(k+1)^d$ edges by keeping only the vertices and edges that are contained in the corresponding sets of the reduced \textsc{$d$-Hitting Set} instance.

Regarding the lower bounds for kernels of {\dPVC}, Dell and Melkebeek~\cite{DellM10} have shown that for \textsc{Vertex Cover} it is not possible to achieve a kernel with $\bigo(k^{2-\varepsilon})$ edges unless coNP is in NP/poly (which would imply a collapse of the polynomial hierarchy). This result extends to {\dPVC} for any $d\geq2$\ifappendix{ (see \autoref{sec:dpvc_lowerbounds} in the appendix for details)}. Therefore, kernels with $\bigo(k^2)$ edges for {\dPVC} are the best we can hope for.

\toappendix{%
\section{Kernelization Lowerbound for \textsc{$d$-Path Vertex Cover}}\label{sec:dpvc_lowerbounds}
In this section, we show that the \textsc{$d$-Path Vertex Cover} problem does not admit a kernel with $\bigo(k^{2-\epsilon})$ edges unless coNP is in NP/poly. For that we use a \emph{polynomial parameter transformation}. For the detailed introduction we refer the reader to Fomin et al. \cite{fomin_lokshtanov_saurabh_zehavi_2019_ppt}. We quote the key definition and theorems here.

\begin{definition}
\emph{A polynomial compression} of a parameterized language $Q \subseteq \Sigma^* \times \mathbb{N}$ into a language $R \subseteq \Sigma^*$ is an algorithm that takes as input an instance $(x,k) \in \Sigma^* \times \mathbb{N}$, works in time polynomial in $|x| + k$, and returns a string $y$ such that:
\begin{enumerate}
\item $|y| \leq p(k)$ for some polynomial $p(\cdot)$, and
\item $y \in R$ if and only if $(x,k) \in Q$.
\end{enumerate}
\end{definition}

Note that a kernelization is a special case of a polynomial compression.

\begin{definition}
Let $P,Q \subseteq \Sigma^* \times \mathbb{N}$ be two parameterized problems. An algorithm $\mathcal{A}$ is called \emph{polynomial parameter transformation (PPT)} from $P$ to $Q$ if, given an instance $(x,k)$ of problem $P$, $\mathcal{A}$ works in polynomial time and outputs an equivalent instance $(\hat{x}, \hat{k})$ of problem $Q$, that is, $(x,k) \in P$ if and only if $(\hat{x}, \hat{k}) \in Q$, such that $\hat{k} \leq p(k)$ for some polynomial $p(\cdot)$.
\end{definition}

The lower bound is based on the following result.

\begin{theorem}[Dell and Melkebeek~\cite{DellM10}]\label{theorem:no_kernel_for_vc}
For any $\varepsilon >0$ the \textsc{Vertex Cover} problem parameterized by the solution size does not admit a polynomial compression with bitsize $\bigo(k^{2-\varepsilon})$, unless coNP $\subseteq$ NP/poly.
\end{theorem}

Our goal is to show that there is a PPT from \textsc{Vertex Cover} to \textsc{$d$-Path Vertex Cover}, both parameterized by the solution size, preserving the value of the parameter.
% such that the polynomial of $k$ in PPT is actually a linear function. Then, a kernel with $\bigo(k^{2-\epsilon})$ edges for \textsc{$d$-Path Vertex Cover} would imply a polynomial compression with bitsize $\bigo(k^{2-\epsilon})$ for \textsc{$d$-Path Vertex Cover} which by \autoref{theorem:ppt} would imply a polynomial compression with bitsize $\bigo(k^{2-\epsilon})$ for \textsc{Vertex Cover} which, unless coNP is in NP/poly, would be a contradiction with the result of Dell and Melkebeek \cite{DellM10}.

\begin{lemma}\label{lemma:dpvc_ppt}
For every $d \ge 3$ there is a PPT that takes as input an instance $(G,k)$ of \textsc{Vertex Cover} and outputs an instance $(G',k)$ of \textsc{$d$-Path Vertex Cover}.
\end{lemma}
\begin{proof}
Given an instance $(G,k)$ of \textsc{Vertex Cover}, we construct a graph $G'$ from $G$ by adding one $(d-2)$-path $P_v$ to each vertex in $v \in V(G)$, i.e., we copy the graph $G$ into $G'$ and for each vertex $v \in V(G)$ we add the path $P_v = (v_1, v_2, \dots, v_{d-2})$ together with an edge $\{v, v_1\}$ to $G'$. The paths $P_u, P_v$ for each two vertices $u,v \in V(G)$ are disjoint.

We now have to show that $(G,k)$ is a YES instance of \textsc{Vertex Cover} if and only if $(G',k)$ is a YES instance of \textsc{$d$-Path Vertex Cover}.

For the first implication, let $S$ be a solution for $(G,k)$. As $S$ is a vertex cover, there are only isolated vertices in $G \setminus S$. As we added only a $(d-2)$-path to each vertex in $G$, there are only $(d-1)$-paths in $G' \setminus S$. Therefore, $S$ is also a solution for $(G', k)$.

For the second implication, let $S'$ be a solution for $(G',k)$. Let $S$ be the set of vertices $v \in V(G)$ for which $S' \cap (\{v\} \cup P_v) \neq \emptyset$, i.e., $S = \{v \mid v \in V(G), S' \cap (\{v\} \cup P_v) \neq \emptyset\}$. We claim that $S$ is a solution for the instance $(G,k)$. Suppose that it is not. Then there is an edge $\{u,v\}$ in $G \setminus S$. But that means that $S' \cap (\{u\} \cup P_u \cup \{v\} \cup P_v) = \emptyset$ and there is a $d$-path $(u,v,v_1,v_2,\dots,v_{d-2})$ in $G' \setminus S'$, which is a contradiction with $S'$ being a solution for $(G', k)$.
\end{proof}

We summarise the original claim in the following corollary.

\begin{corollary}
For any $d \ge 3$ and any $\varepsilon >0$ the \textsc{$d$-Path Vertex Cover} problem parameterized by the solution size does not admit a polynomial compression (in particular a kernel) with bitsize $\bigo(k^{2-\varepsilon})$, unless coNP $\subseteq$ NP/poly.
\end{corollary}
\begin{proof}
Suppose that \textsc{$d$-Path Vertex Cover} admits a polynomial compression $\mathcal{K}$ with bitsize $\bigo(k^{2-\epsilon})$.
Let $(G,k)$ be an instance of \textsc{Vertex Cover}. By \autoref{lemma:dpvc_ppt} there is a PPT $\mathcal{A}$ which takes the instance $(G,k)$ and outputs an equivalent instance $(G',k)$ of \textsc{$d$-Path Vertex Cover}.
We now run the compression $\mathcal{K}$ on $(G',k)$ to obtain an equivalent instance $y$ with bitsize $\bigo(k^{2-\epsilon})$.
This way we obtain a polynomial compression with bitsize $\bigo(k^{2-\epsilon})$ for \textsc{Vertex Cover}.
By \autoref{theorem:no_kernel_for_vc}, this implies coNP $\subseteq$ NP/poly.
\end{proof}
}%

The current best kernels known are a kernel for \textsc{Vertex Cover} with $2k - c\log{k}$ vertices for any fixed constant $c$ due to Lampis~\cite{Lampis11} and a kernel with $5k$ vertices for \textsc{$3$-PVC} due to Xiao and Kou~\cite{XiaoK17}. No specific kernels are known for \textsc{$d$-PVC} with $d \ge 4$, except for those inherited from $d$-\textsc{Hitting Set}.

Dell and Marx~\cite{DellM18} recently studied kernels for the related \textsc{$d$-Path Packing} problem, which also inspired our work.

\paragraph*{Our contribution.}
We give kernels with $\bigo(k^2)$ edges for \textsc{4-PVC} and \textsc{5-PVC} (asymptotically optimal, unless coNP $\subseteq$ NP/poly).
Furthermore, for the general case, we give a kernel for {\dPVC} for any $d\geq6$ with $\bigo(k^4d^{2d+9})$ edges.
% Also, inspired by the recent work of Dell and Marx~\cite{DellM18} on the related \textsc{$d$-Path Packing} problem,  .

\section{Preliminaries}\label{sec:prelims}
\toappendix{\section{Additional Material to \autoref{sec:prelims}}}%
% We use the $\ostar$ \textit{notation} as described by Fomin and Kratsch~\cite{FominK10}, which is a~modification of the big-$\bigo$ notation suppressing all polynomially bounded factors.\os{We do not, not so far. Also useless in preliminaries.}
We use the notations related to parameterized complexity as described by Cygan et al.~\cite{CyganFKLMPPS15}.  We consider simple and undirected graphs unless otherwise stated. For a graph $G$, we use $V(G)$ to denote the vertex set of $G$ and $E(G)$ to denote the edge set of $G$. By $G[X]$ we denote the subgraph of~$G$ induced by vertices of $X \subseteq V(G)$. By $N(v)$ we denote the~set of neighbors of $v \in V(G)$ in~$G$. Analogically, $N(X) = \bigcup_{x \in X} N(x)\setminus X$ denotes the~set of neighbors of vertices in $X \subseteq V(G)$. The degree of vertex~$v$ is denoted by $\deg(v) = |N(v)|$. For simplicity, we write $G \setminus v$ for $v \in V(G)$ and $G \setminus X$  for $X \subseteq V(G)$ as shorthands for $G[V(G)\setminus \{v\}]$ and $G[V(G)\setminus X]$, respectively.
% we only used $G-v$ or $G-X$ twice in the paper

A~\textit{$d$-path} (also denoted by $P_d$), denoted as an ordered~$d$-tuple $(p_1, p_2, \ldots, p_d)$, is a~path on $d$~vertices $\{p_1, p_2, \ldots, p_d\}$. A~\textit{$d$-path free} graph is a~graph that does not contain a~$d$-path as a~subgraph (the $d$-path needs not to be induced).
The \emph{length} of a path $P$ is the number of edges in $P$, in particular, the length of a $d$-path $P_d$ is $d-1$.

The \textsc{$d$-Path Vertex Cover} problem is formally defined as follows:

\vspace{2mm}
\noindent
\begin{tabularx}{\textwidth}{|l|X|}
  \hline
\multicolumn{2}{|l|}{\textsc{$d$-Path Vertex Cover, $d$-PVC}} \\ \hline
  \textsc{Input}: & A~graph $G=(V,E)$, a non-negative integer $k$. \\
  \textsc{Output}: & A~set $S \subseteq V$, such that $|S| \leq k$ and $G \setminus S$ is a~$P_d$-free graph. \\
  \hline
\end{tabularx}
\vspace{2mm}

% \begin{definition}
A $d$-path packing $\PP$ of size $l$ in a graph $G$ is a collection of $l$ vertex disjoint $d$-paths in the graph~$G$. We use $V(\PP)$ to denote the union of the vertex sets of the $d$-paths in the packing $\PP$. For rest of the graph theory notations we refer to Diestel~\cite{diestel}.
%\end{definition}

For a positive integer $i$, we will use $[i]$ to denote the set $\{1,2,\dots,i\}$.

\begin{proposition}[\appmark]\ifappendix{\footnote{Proofs of (correctness of) items marked with (\appmark{}) can be found in the Appendix.}}
\label{prop:greedy_packing}
For a given graph $G$ and an integer $k$, there is an algorithm which either correctly answers whether $G$ has a $d$-path vertex cover of size at most $k$,
or finds an inclusion-wise maximal $d$-path packing $\PP$ of size at most $k$ in $\bigo(2.619^d d k n\log{n})$ time.
\end{proposition}

\toappendix{%
% \begin{proof}
\begin{proof}[Proof of \autoref{prop:greedy_packing}]
The algorithm of Fomin et al.~\cite{rep-sets} decides whether there is a $d$-path in $G$ in time $\bigo(2.619^dn\log{n})$.
By a framework of Bj{\"{o}}rklund et al.~\cite{BjorklundKK14} the algorithm can turned into one that actually finds a $d$-path in time $\bigo(2.619^d d n\log{n})$.
We then call the algorithm to find a $d$-path at most $k+1$ times to construct an appropriate answer (see \autoref{algorithm:greedy_packing}.
Therefore, the total running time of our greedy algorithm is $\bigo(2.619^d d k n\log{n})$.

\begin{algorithm}[H]
 \KwIn{A graph $G$, a non-negative integer $k$.}
 \KwOut{An inclusion-wise maximal $d$-path packing $\PP$, or a decision on whether $G$ has a $d$-path vertex cover of size $\le k$}
 $\PP \gets \emptyset$\\
 \While{There is a $d$-path $P_d$ in $G \setminus V(\PP)$ and $|\PP| \leq k$ }{
     Find and add the $d$-path $P_d$ to the packing $\PP$.
 }%while
 \lIf (\tcp*[f]{there is a $d$-path vertex cover of size $\le k$}){$|\PP|=0$}{ answer YES}
 \lIf (\tcp*[f]{there is no $d$-path vertex cover of size $\le k$}){$|\PP| \geq k+1$}{answer NO}
 \KwRet{$\PP$}
 \caption{Greedy $d$-path packing algorithm.}
 \label{algorithm:greedy_packing}
\end{algorithm}

To see that the algorithm answers correctly, simply observe that for each $d$-path in $\PP$ there must be at least one vertex in any solution for $G$. Therefore if $|\PP| \geq k+1$, then any solution will use at least $k+1$ vertices and therefore the answer is NO. In the case where $|\PP| = 0$, the graph $G$ is already $P_d$-free graph, and the answer is YES. Finally, if the first two cases do not apply, we return an inclusion-wise maximal $d$-path packing $\PP$ with at most $k$ paths. If the packing would not be inclusion-wise maximal, the algorithm would simply find a larger packing.
\end{proof}
}% to appendix

\section{General Reduction Rules}\label{sec:general}
\toappendix{\section{Additional Material to \autoref{sec:general}}}%
Let us start with reduction rules that apply to $d$-\PVC{} for most values of $d$.
Assume that we are working with an instance $(G=(V,E), k)$ of $d$-\PVC{} for some $d \ge 4$.
We start with a reduction rule whose correctness is immediate.

\begin{rrule}\label{rule:component}
 If there is a connected component $C$ in $G$ which does not contain a $P_d$, then remove $C$.
\end{rrule}

The next rule allows us to get rid of multiple degree-one vertices adjacent to a single vertex.
\begin{rrule}[\appmark]
\label{rule:degree-one}
Let there be three distinct vertices $v, x, y \in V$ such that $N(x) =  N(y) = \{v\}$. We reduce the instance by deleting the vertex $x$.
\end{rrule}

\toappendix{%
% \noindent
% \textbf{Reduction Rule} \ref{rule:degree-one}.
% {\em Let there be three distinct vertices $v, x, y \in V$ such that $N(x) =  N(y) = \{v\}$. We reduce the instance by deleting the vertex $x$.}
\begin{proof}[Proof of Correctness of \autoref{rule:degree-one}]
%\begin{proof}
Let $(G,k)$ be the original instance and $(G',k)$ be the reduced one, i.e., $G' = G \setminus x$.
Since $G'$ is a subgraph of $G$, if $S$ is a solution for $G$, then $S\setminus \{x\}$ is a solution for $G'$.
Hence we focus on the other direction.

Let $S'$ be a solution for $G'$.
If $S'$ is a solution for $G$, then we are done.
Suppose it is not.
This means that there is a $d$-path in $G \setminus S'$.
Such a path necessarily uses vertex $x$.
Also observe, that $v \notin S'$, otherwise $x$ would be isolated in $G \setminus S'$ and, hence, could not be part of a $d$-path.

On one hand, suppose that $y \in S'$.
Let $S = (S' \setminus \{y\}) \cup \{v\}$.
We claim that $S$ is a solution for $G'$ and $G$. Indeed, any $d$-path which uses vertex $y$ must go through vertex $v$ as $N(y) = \{v\}$, which means that any such path is covered by the set $S$. The same argument works in the case of vertex $x$ and therefore the set $S$ is also a solution for graph $G$. Lastly, we have that $|S| \leq |S'|$, as we are only switching vertex $y$ for $v$.

On the other hand, suppose that $y \notin S'$. Let the $d$-path in $G \setminus S'$ be $(x,v,u_3,\ldots,u_d)$. We have that $\{x,v,u_3,\ldots,u_d\} \cap S' = \emptyset$. But that means, that there is a $d$-path $(y,v,u_3,\ldots,u_d)$ in $G' \setminus S'$ contradicting the fact, that $S'$ is a solution for $G'$.
\end{proof}
}% to appendix

\section{High Degree Reduction Rule for $4$-\PVC{} and $5$-\PVC{}}\label{sec:high_degree}
\toappendix{\section{Additional Material to \autoref{sec:high_degree}}}%
In this section, we are going to introduce the reduction rules which are applicable to both $4$-\PVC{} and $5$-\PVC{} instances.
We assume that we are working with a $d$-\textsc{PVC} instance $(G=(V,E), k)$ for $d \in \{4,5\}$ which is reduced by exhaustively employing \autoref{rule:degree-one}.

Our aim is to show that the degree of each vertex can be reduced to linear in the parameter.
First assume that there is a large matching in the neighborhood of some vertex $v$.
We call a matching $\M$ in $G$ \emph{adjacent} to vertex $v$, if it is a matching in $G\setminus v$ and for each edge $\{a_i, b_i\}\in \M$ at least one of its vertices, say $a_i$, is adjacent to $v$ in $G$.

\begin{rrule}[\appmark]
\label{rule:matching}
If $v$ is a vertex and $\M$ a matching adjacent to $v$ of size $|\M| \ge k+2$, then delete $v$ and decrease $k$ by $1$.
\end{rrule}

\toappendix{%
% \begin{proof}[of Correctness]
\begin{proof}[Proof of Correctness of \autoref{rule:matching}]
Let $M$ be the set of vertices covered by matching $\M$.
If $S'$ is a solution for $G'=G\setminus\{v\}$ of size at most $k'=k-1$, then $S' \cup \{v\}$ is a solution for $G$ of size at most $k$.
If $S$ is a solution for $G$ of size at most $k$ that contains $v$, then $S \setminus \{v\}$ is a solution for $G'$ of size at most $k'$.
Suppose that there is a solution $S$ for $G$ which does not use the vertex $v$.
The solution $S$ deletes at most $k$ vertices from $M$ which leaves us with at least two distinct uncovered edges $\{a_i, b_i\}, \{a_j, b_j\}$ in $\M \setminus S$. But then we have a 5-path $(b_i,a_i, v, a_j, b_j)$ in $G \setminus S$ which is a contradiction with $S$ being a solution for~$G$.
\end{proof}
}% to appendix

To exhaustively apply \autoref{rule:matching}, we need to find for each $v \in V$ a largest matching adjacent to $v$.
This can be done as follows. Let $A=N(v)$ and $B=N(A) \setminus \{v\}$.
Let $G_v$ be the graph obtained from $G[A \cup B]$ by removing edges with both endpoints in $B$.
It is easy to observe, that each matching adjacent to $v$ is also a matching in $G_v$ and vice-versa.
Hence, it suffices to find a largest matching in $G_v$, which can done in polynomial time~\cite{Edmonds65}.

Therefore, we further assume that the instance is reduced with respect to \autoref{rule:matching}.
% Let $v \in V$ be a vertex of $G$ with $\deg(v) \geq 7k + 8$.
% We greedily construct\todo{explain} a maximal matching $\M$ of $m$ edges $\{a_1, b_1\}, \{a_2, b_2\}, \ldots, \{a_m, b_m\}$ such that $\{a_1, a_2, \ldots, a_m\} \subseteq N(v)$.
We fix a vertex $v$ and find a largest matching $\M$ adjacent to it by the above algorithm.
Let $M$ be the set of vertices covered by matching $\M$ and $m = |\M|$.
Since the instance is reduced, we know that $m \le k+1$.
Let $X = N(v) \setminus M$.
We refer the reader to the \ifappendix{\autoref{fig:setting} (Appendix) or }\autoref{fig:setting2} for overview of our setting.

\toappendix{
\begin{figure}[!h]
\begin{center}
\begin{tikzpicture}[every node/.style={draw, fill, circle, inner sep=1.5pt}]
\node[label=above:$v$, red] (v) at (-1,4.5) {};

% matching M
\begin{scope}[shift={(0,0)}]
\foreach \i in {1,2,3}{
  \node[label=above :$a_{\i}$] (a\i) at (\i,2) {};
  \node[label=below :$b_{\i}$] (b\i) at (\i,1) {};
  \draw[thick] (a\i) to (b\i);
  \draw[thick, bend right=10, vertexVorangeEdge] (v) to (a\i);
  \draw[dashed, thick, bend right=10, vertexVorangeEdge] (v) to (b\i);
}
  \node[label=above :$a_{m}$] (a4) at (5,2) {};
  \node[label=below :$b_{m}$] (b4) at (5,1) {};
  \draw[thick] (a4) to (b4);
  \draw[thick, bend right=5, vertexVorangeEdge] (v) to (a4);
  \draw[dashed, thick, bend right=10, vertexVorangeEdge] (v) to (b4);
  \node[draw=none, fill=none] at (4,1.5) {$\cdots$};

\begin{scope}[on background layer]
    \node[rectangle,color=matchingMBoxDraw,draw,fill=matchingMBoxFill,dashed,inner sep=1cm,rounded corners,fit=(a1)(b4),label={left:$M$}] (Mfitter) {};
\end{scope}
 
\foreach \i in {1,2,3,4}{ 
\foreach \j in {1,2,3,4}{ 
 \draw[dashed, black!80!white] (a\i) to (b\j);
   
  \ifthenelse{\i<\j}{\draw[dashed, black!80!white, bend left=20] (a\i) to (a\j);}{}
  \ifthenelse{\i<\j}{\draw[dashed, black!80!white, bend right=20] (b\i) to (b\j);}{}

}
} 
 
\end{scope}
  
% set X
\begin{scope}[shift={(0,4)}]
\foreach \i in {1,2,...,6}{
  \node[setXvertex] (x\i) at (\i,0) {}; 
  \draw[thick, bend left=5, vertexVorangeEdge] (v) to (x\i);
 }
 \node[draw=none, fill=none] at (7,0) {$\cdots$};   
 \node[setXvertex] (x7) at (8,0) {}; 
  \draw[thick, bend left=5, vertexVorangeEdge] (v) to (x7);

\foreach \i in {1,2,3,4}{ 
\foreach \j in {1,2,...,7}{
  \draw[dashed, setXvertex!50!white] (a\i) to (x\j);
  \draw[dashed, setXvertex!50!white] (b\i) to (x\j);  
}  
}

%todo: edges inside M
\begin{scope}[on background layer]
    \node[rectangle,color=setXBoxDraw,draw,fill=setXBoxFill,dashed,inner sep=0.5cm,rounded corners,fit=(x1)(x7),label={right:$X$}] (Xfitter) {};
  \end{scope}

\end{scope}

% redraw the coords so they are visible
\foreach \i in {1,2,3}{
\node[label=above :$a_{\i}$] (a{\i}label) at (a\i) {};
\node[label=below:$b_{\i}$] (b{\i}label) at (b\i) {};
}
 \node[label=above :$a_{m}$] (a4label) at (a4) {};
  \node[label=below :$b_{m}$] (b4label) at (b4) {};

\end{tikzpicture}
\end{center}
\caption{An overview of the setting of vertex $v$ and sets $M, X$ for the high degree reduction rule.}
\label{fig:setting}
\end{figure}
}%to appendix

\begin{observation}[\appmark]
\label{obs:first}
For each $x \in X$ we have $N(x) \setminus \{v\} \subseteq M$.
\end{observation}

\toappendix{%
% \noindent
% \textbf{Observation} \ref{obs:first}.
% {\em For each $x \in X$ we have $N(x) \setminus \{v\} \subseteq M$.}
%
\begin{proof}[Proof of \autoref{obs:first}]
% \begin{proof}
Suppose that there is a vertex $x \in X$ and a vertex $u \in N(x) \setminus (M \cup \{v\})$. We construct a new larger matching $\M' = \M \cup \{\{x, u\}\}$ adjacent to $v$, contradicting the assumption that $\M$ is a largest matching.
\end{proof}
}% to appendix

\begin{observation}[\appmark]
\label{obs:noOppositeEnds}
No two distinct vertices $x,y \in X$ are connected to the opposite endpoints of a single edge $\{a_i, b_i\}$ in $\M$.
%Otherwise, we can increase the size of the matching $\M$.
\end{observation}

\toappendix{%
% \noindent
% \textbf{Observation} \ref{obs:noOppositeEnds}.
% {\em No two distinct vertices $x,y \in X$ are connected to the opposite endpoints of a single edge $\{a_i, b_i\}$ in $\M$.}
%
\begin{proof}[Proof of \autoref{obs:noOppositeEnds}]
% \begin{proof}
Without loss of generality, assume that $a_i \in N(x)$ and $b_i \in N(y)$.
For an illustration see \autoref{fig:noOppositeEnds}.
We construct a new larger matching $\M' = (\M \setminus \{a_i, b_i\}) \cup \{\{x, a_i\}, \{y, b_i\}\}$ adjacent to $v$,
contradicting the assumption that $\M$ is a largest matching.
\end{proof}
\begin{figure}[ht!]
\begin{center}
\begin{tikzpicture}[every node/.style={draw, fill, circle, inner sep=1.5pt}]

% left figure
\begin{scope}[shift={(-1,0)}]
\node[label=above:$v$, red] (v) at (2,3) {};

% matching M
\node[label=above left:$a_{i}$] (a) at (1,2) {};
\node[label=below left:$b_{i}$] (b) at (1,1) {};
\draw[very thick] (a) to (b);
\draw[thick, bend right=20, vertexVorangeEdge] (v) to (a);
\draw[dashed, thick, bend right=20, vertexVorangeEdge] (v) to (b);

\begin{scope}[on background layer]
    \node[rectangle,color=matchingMBoxDraw,draw,fill=matchingMBoxFill,dashed,inner sep=0.5cm,rounded corners,fit=(a)(b),label={left:$M$}] (Mfitter) {};
\end{scope}

\node[setXvertex, label=above right:$x$] (x) at (3,2) {};
\node[setXvertex, label=above right:$y$] (y) at (4,2) {};
\draw[thick, bend left=20, vertexVorangeEdge] (v) to (x);
\draw[thick, bend left=20, vertexVorangeEdge] (v) to (y);
\draw[setXvertex, thick] (a) to (x);
\draw[setXvertex, thick, bend right] (b) to (y);
\draw[setXvertex, dashed] (b) to (x);
\draw[setXvertex, dashed, bend right] (a) to (y);

\begin{scope}[on background layer]
    \node[rectangle,color=setXBoxDraw,draw,fill=setXBoxFill,dashed,inner sep=0.5cm,rounded corners,fit=(x)(y),label={right:$X$}] (Xfitter) {};
  \end{scope}

\end{scope}

\draw[line width=2pt,-{Stealth[round]}] (5,2) to (6,2);

% right figure

\begin{scope}[shift={(7,0)}]
\node[label=above:$v$, red] (v) at (2,3) {};
\node[label=above left:$a_{i}$] (a) at (1,2) {};
\node[label=below left:$b_{i}$] (b) at (1,1) {};
\draw[dashed] (a) to (b);
\draw[thick, bend right=20, vertexVorangeEdge] (v) to (a);
\draw[dashed, thick, bend right=20, vertexVorangeEdge] (v) to (b);
\node[label=above right:$x$] (x) at (3,2) {};
\node[label=above right:$y$] (y) at (4,2) {};
\draw[thick, bend left=20, vertexVorangeEdge] (v) to (x);
\draw[thick, bend left=20, vertexVorangeEdge] (v) to (y);
\draw[thick] (a) to (x);
\draw[thick, bend right] (b) to (y);
\draw[dashed] (b) to (x);
\draw[dashed, bend right] (a) to (y);

\begin{scope}[on background layer]
    \node[rectangle,color=matchingMBoxDraw,draw,fill=matchingMBoxFill,dashed,inner sep=0.5cm,rounded corners,fit=(a)(b)(y),label={left:$M'$}] (Mfitter) {};
\end{scope}

\end{scope}

\end{tikzpicture}
\end{center}
\caption{An illustration of the situation in \autoref{obs:noOppositeEnds}.}
\label{fig:noOppositeEnds}
\end{figure}
}% to appendix

\begin{observation}[\appmark]\label{obs:onlyOneDouble}
If there is a vertex $x \in X$ such that for some edge $\{a_i, b_i\}$ in the matching $\M$ we have that $\{a_i, b_i\} \subseteq N(x)$, then $N(\{a_i,b_i\}) \cap X = \{x\}$.
% there is no other vertex $y \in X$ such that $\{a_i, b_i\} \cap N(y) \neq \emptyset$.
%Otherwise, we can increase the size of the matching $\M$.
\end{observation}

\toappendix{%
% \noindent
% \textbf{Observation} \ref{obs:onlyOneDouble}.
% {\em If there is a vertex $x \in X$ such that for some edge $\{a_i, b_i\}$ in the matching $\M$ we have that $\{a_i, b_i\} \subseteq N(x)$, then $N(\{a_i,b_i\}) \cap X = \{x\}$}
%
\begin{proof}[Proof of \autoref{obs:onlyOneDouble}]
% \begin{proof}
Let $y \in ((N(\{a_i,b_i\}) \cap X) \setminus \{x\})$.
If $a_i \in N(y)$, then, as $b_i \in N(x)$ we get a contradiction with \autoref{obs:noOppositeEnds}. Similarly if $b_i \in N(y)$, we can derive a contradiction.
\end{proof}
}% to appendix

% \begin{observation}\label{obs:onlyOneOneDeg}
% There is at most one vertex $x \in X$ such that $N(x) = \{v\}$.
% \end{observation}
%
% \begin{proof}
% Follows directly from the fact that the problem instance $(G=(V,E), k)$ is reduced with respect to the degree-one reduction rule.
% \end{proof}

We now partition the set $X$ into three sets.
Let $X_2$ be the set of vertices such that for each $x \in X_2$ we have some edge $\{a_i, b_i\}$ in the matching $\M$ such that $\{a_i, b_i\} \subseteq N(x)$.
Let $X_0$ be the vertices such that for each $x \in X_0$ we have that $N(x) = \{v\}$.
Note, that $X_0$ contains at most one vertex due to \autoref{rule:degree-one} being exhaustively applied.
Lastly, let $X_1 = X \setminus (X_2 \cup X_0)$ be the rest of the vertices in $X$.
See \autoref{fig:setting2} for an illustration of the sets $X_2$, $X_0$, and~$X_1$.

\begin{figure}[!h]
\begin{center}
\begin{tikzpicture}[every node/.style={draw, fill, circle, inner sep=1.5pt}]
\node[label=above:$v$, red] (v) at (2,4.5) {};

% matching M
\begin{scope}[shift={(0,0)}]
\foreach \i in {1,2,3}{
  \node[label=above left:$a_{\i}$] (a\i) at (\i,2) {};
  \node[label=below left:$b_{\i}$] (b\i) at (\i,1) {};
  \draw[thick] (a\i) to (b\i);
  \draw[thick, bend right=20, vertexVorangeEdge] (v) to (a\i);
  \draw[dashed, thick, bend right=20, vertexVorangeEdge] (v) to (b\i);
}
  \node[label=above left:$a_{m}$] (a4) at (5,2) {};
  \node[label=right:$b_{m}$] (b4) at (5,1) {};
  \draw[thick] (a4) to (b4);
  \draw[thick, bend right=20, vertexVorangeEdge] (v) to (a4);
  \draw[dashed, thick, bend right=20, vertexVorangeEdge] (v) to (b4);
  \node[draw=none, fill=none] at (4,1.5) {$\cdots$};

\begin{scope}[on background layer]
    \node[rectangle,color=matchingMBoxDraw,draw,fill=matchingMBoxFill,dashed,inner sep=1cm,rounded corners,fit=(a1)(b4),label={left:$M$}] (Mfitter) {};
\end{scope}

% M_1 set  
\begin{scope}[on background layer]

\draw[rounded corners,setM1BoxDraw,fill=setM1BoxFill,rotate around={26.57:(3,1)}] (3-0.5,1-0.5) rectangle (5+0.6,1+0.5);
\coordinate[label=$M_1$, setM1BoxDraw] (M1Label) at (3,-0.2);

\end{scope}   
\end{scope}
  
% set X
\begin{scope}[shift={(3,4)}]
\foreach \i in {1,2,...,5}{
  \node[setXvertex] (x\i) at (\i,0) {}; 
  \draw[thick, bend left=10, vertexVorangeEdge] (v) to (x\i);
 }
 \node[draw=none, fill=none] at (6,0) {$\cdots$};   
 \node[setXvertex] (x6) at (7,0) {}; 
  \draw[thick, bend left=10, vertexVorangeEdge] (v) to (x6); 
 \node[setXvertex] (x7) at (8,0) {}; 
  \draw[thick, bend left=10, vertexVorangeEdge] (v) to (x7); 
  
%todo: edges inside M
\begin{scope}[on background layer]
    \node[rectangle,color=setXBoxDraw,draw,fill=setXBoxFill,dashed,inner sep=0.5cm,rounded corners,fit=(x1)(x7),label={right:$X$}] (Xfitter) {};
  \end{scope}

% X_2 set  
\draw[thick, setX2edge, bend left=0] (a1) to[out=30,in=-150] (x1);
\draw[thick, setX2edge, bend left=0] (b1) to[out=30,in=-150] (x1);
\draw[thick, setX2edge, bend left=0] (a2) to[out=0,in=-140] (x2);
\draw[thick, setX2edge, bend left=0] (b2) to[out=30,in=-140] (x2);
\begin{scope}[on background layer]
\node[rectangle, color=setX2BoxDraw, fill=setX2BoxFill,
    dashed,inner sep=0.3cm,rounded corners, fit=(x1) (x2),
    label={left:$X_2$}
    ] (X2fitter) {};
\end{scope} 

% X_0 set  
\begin{scope}[on background layer]
\node[rectangle, color=setX0BoxDraw, fill=setX0BoxFill,
    dashed,inner sep=0.3cm,rounded corners, fit=(x7),
    label={below:$X_0$}
    ] (X0fitter) {};
\end{scope} 

% X_1 set  
\draw[thick, setX1edge, bend left=0] (x3) to[out=60,in=150] (a4);
\draw[thick, setX1edge, bend left=0] (x4) to[out=40,in=160] (a4);
\draw[thick, setX1edge, bend left=0] (x5) to[out=60,in=120] (b3);
\draw[thick, setX1edge, bend left=0] (x5) to[out=30,in=130] (a4);
\draw[thick, setX1edge, bend left=0] (x6) to[out=50,in=120] (b3);
\begin{scope}[on background layer]
\node[rectangle, color=setX1BoxDraw, fill=setX1BoxFill,
    dashed,inner sep=0.3cm,rounded corners, fit=(x3) (x6),
    label={280:$X_1$}
    ] (X1fitter) {};
\end{scope} 

\end{scope}

\end{tikzpicture}
\end{center}
\caption{An overview of the definitions of sets $X_0$, $X_1$, $X_2$, and $M_1$.}
\label{fig:setting2}
\end{figure}

\begin{observation}[\appmark]
 \label{obs:X_1_size}
 If the vertex $v$ has degree at least $(d+2)(k+1)+1$, then $|X_1| \geq (d-1)(k+1)$.
\end{observation}

\toappendix{%
\begin{proof}[Proof of \autoref{obs:X_1_size}]
Due to \autoref{obs:onlyOneDouble} there can be at most one vertex in $X_2$ per each edge in $\M$ and, since the instance is reduced with respect to \autoref{rule:matching}, we have $|\M| \leq k+1$. Hence $|X_2| \leq k+1$.
Since the instance is reduced with respect to \autoref{rule:degree-one}, we have $|X_0| \leq 1$.

Suppose that the vertex $v$ has degree at least $(d+2)(k+1)+1$.
There are at most $2(k+1)$ edges between $v$ and $\M$, which leaves at least $d(k+1)+1$ edges for the set $X$.
As the set $|X_2|$ contains at most $k+1$ vertices and the set $|X_0|$ contains at most one vertex, there can be at most $k+2$ edges between $v$ and sets $X_0$ and $X_2$.
This leaves us with at least $(d-1)(k+1)$ edges between $v$ and $X_1$, which means that $|X_1| \geq (d-1)(k+1)$.
\end{proof}
}%to appendix

Now we focus on the edges between $X_1$ and $M$.
By \autoref{obs:noOppositeEnds}, for each edge $\{a_i, b_i\}$ in $\M$ we have that the vertices in $X_1$ may be adjacent to at most one vertex of such edge, i.e. $|\{a_i, b_i\} \cap N(X_1)| \leq 1$.
Letting $M_1 = M \cap N(X_1)$ we have $|M_1| \leq k+1$.

We are now ready to employ the Expansion Lemma. We use the version of Fomin~et~al.~\cite{FominLMPS16ExpansionLemma}, which is a generalization of the original results by Prieto~\cite[Corollary 8.1]{Prieto05ExpansionLemma} and Thomassé~\cite[Theorem 2.3]{Thomasse10ExpansionLemma}.

% taken from the squirell book
\begin{definition}
Let $G$ be a bipartite graph with vertex bipartition $(A,B)$. A set of edges $Q \subseteq E(G)$ is called a $q$-expansion, $q \geq 1$, of $A$ into $B$ if every vertex of $A$ is incident with exactly $q$ edges of $Q$, and $Q$ saturates exactly $q|A|$ vertices in~$B$.
\end{definition}

\begin{lemma}[Expansion Lemma; {Fomin et al.~\cite{FominLMPS16ExpansionLemma}}]\label{lemma:expansionLemma} %\cite{fomin_lokshtanov_saurabh_zehavi_2019}
Let $q$ be a positive integer, and $G$ be a bipartite graph with bipartition $(A,B)$ such that $|B| \geq q|A|$, and there are no isolated vertices in $B$. Then, there exists nonempty $A' \subseteq A$ and $B' \subseteq B$ such that $A'$ has a $q$-expansion into $B'$ and $N(B') \subseteq A'$.
Moreover, the sets $A', B'$ and the $q$-expansion can be found in polynomial time.
\end{lemma}

\begin{observation}[\appmark]
\label{obs:expansion}
There exist non-empty subsets $M' \subseteq M_1$ and $X' \subseteq X_1$ such that there is a $(d-1)$-expansion $Q'$ from $M'$ into $X'$ and $N(X') \subseteq M' \cup \{v\}$.
\end{observation}

\toappendix{%
\begin{proof}[Proof of \autoref{obs:expansion}]
% \begin{proof}
Let us consider the bipartite subgraph of $G$ formed  by vertex sets $M_1$ and $X_1$ and all edges of $G$ in between them.
As the size of $M_1$ is at most $k+1$ and the size of $X_1$ is at least $(d-1)(k+1)$, we have that $|X_1| \geq (d-1)|M_1|$. Further, there are no isolated vertices in $X_1$, as the isolated vertices would be in the set $X_0$.

The conditions of the \autoref{lemma:expansionLemma} are satisfied and the existence of the subsets $M', X'$ directly follows.
\end{proof}
}% to appendix

\begin{figure}[!h]
\begin{center}
\begin{tikzpicture}[every node/.style={draw, fill, circle, inner sep=1.5pt}]
\node[label=above:$v$, red] (v) at (2,4.5) {};

% matching M
\begin{scope}[shift={(0,0)}]
\foreach \i in {1,2,4,5}{
  \node[label=above left:$$] (a\i) at (\i,2) {};
  \node[label=below left:$$] (b\i) at (\i,1) {};
  \draw[thick] (a\i) to (b\i);
  %\draw[thick, bend right=20, vertexVorangeEdgeDim] (v) to (a\i);
  %\draw[dashed, thick, bend right=20, vertexVorangeEdgeDim] (v) to (b\i);
}
  %\node[label=above left:$a_{m}$] (a4) at (5,2) {};
  %\node[label=below left:$b_{m}$] (b4) at (5,1) {};
  %\draw[thick] (a4) to (b4);
  %\draw[thick, bend right=20, vertexVorangeEdgeDim] (v) to (a4);
  %\draw[dashed, thick, bend right=20, vertexVorangeEdgeDim] (v) to (b4);
  \node[draw=none, fill=none] at (3,1.5) {$\cdots$};

\begin{scope}[on background layer]
    \node[rectangle,color=setM1BoxDraw,draw,fill=setM1BoxFill,dashed,inner sep=0.5cm,rounded corners,fit=(a1)(b5),label={left:$M_1$}] (Mfitter) {};
\end{scope}

% M_1 set  
\begin{scope}[on background layer]
\draw[rounded corners,setMprimeBoxDraw,fill=setMprimeBoxFill,rotate around={-45:(4,2)}] (3.5,2.5) rectangle (4+1.41+0.5,1.5);
\coordinate[label=$M'$] (M1Label) at (6,0.5);

\end{scope}   
\end{scope}
  
% set X
\begin{scope}[shift={(0,3.5)}]
 \node[setXvertex, label={right:$x$}] (x1) at (1,0) {}; 
  \draw[very thick, bend left=0, red] (v) to (x1);
\foreach \i in {2,...,9}{
  \node[setXvertex] (x\i) at (\i,0) {}; 
  \draw[thick, bend left=10, vertexVorangeEdge] (v) to (x\i);
 }
 \node[draw=none, fill=none] at (10,0) {$\cdots$};   
 \node[setXvertex] (x10) at (11,0) {}; 
  \draw[thick, bend left=10, vertexVorangeEdge] (v) to (x10);

\begin{scope}[on background layer]
    \node[rectangle,color=setX1BoxDraw,draw,fill=setX1BoxFill,dashed,inner sep=0.5cm,rounded corners,fit=(x1)(x10),label={right:$X_1$}] (Xfitter) {};
  \end{scope}

% X_1 set  
\foreach \i in {1,2,3,4} {
    \draw[very thick, setXprimeEdge] (x\i) to (a4);
}
\foreach \i in {5,6,7,8} {
    \draw[very thick, setXprimeEdge] (x\i) to[out=270,in=20] (b5);
}

\coordinate[label=$Q$] (QLabel) at (6.2,-1.5);

\begin{scope}[on background layer]
\node[rectangle, color=setXprimeBoxDraw, fill=setXprimeBoxFill,
    dashed,inner sep=0.3cm,rounded corners, fit=(x1) (x9),
    label={[shift={(3.2,-1.5)}]:$X'$}
    ] (X1fitter) {};
\end{scope} 

\end{scope}

\end{tikzpicture}
\end{center}
\caption{A graphical interpretation of applying the Expansion Lemma to our setting.}
\label{fig:setting3}
\end{figure}

We refer the reader to \autoref{fig:setting3} for a graphical interpretation of the situation guaranteed by \autoref{obs:expansion}.
Now, let us focus on the sets $M'$ and $X'$ and the way they are connected with vertex $v$. We are going to show that some edge between $v$ and $X'$ is now redundant.

\begin{rrule}\label{rule:expansion}
Let $v$ be a vertex of degree at least $(d+2)(k+1)+1$.
Let $\M$ be a largest matching adjacent to $v$ and $M$ be the set of vertices covered by $\M$.
Let $X_1\subseteq N(v) \setminus M$ be the set of vertices $x$ with $N(x) \cap M \neq \emptyset$ and $|N(x) \cap \{a_i,b_i\}|\le 1$ for each $\{a_i,b_i\} \in \M$.
Let $M_1=M\cap N(X_1)$.
Let the non-empty subsets $M' \subseteq M_1$ and $X' \subseteq X_1$ be the sets with the $(d-1)$-expansion $Q$ from $M'$ into $X'$ and such that $N(X') \subseteq M' \cup \{v\}$.
Let $x \in X'$. Reduce the instance by deleting the edge $\{x, v\}$.
\end{rrule}

\begin{proof}[Proof of Correctness.]
Let $(G=(V,E),k)$ be the original instance and $(G'=(V,E'),k)$ the reduced one.
For each vertex $m \in M'$ let $Q_{m}$ be the set of vertices of $X'$ incident to $m$ in the $(d-1)$-expansion $Q$.
Since $G'$ is a subgraph of $G$, if $S$ is a solution for $G$, then $S$ is also a solution for $G'$.
Hence we will concetrate on the other direction.

Suppose that $S'$ is a solution for the reduced instance.
If it is also a solution for the original one, then we are done.
Suppose it is not,
i.e., there is a $P_d$ in $G \setminus S'$. This $P_d$ contains the edge $\{x, v\}$, otherwise it would be also present in $G' \setminus S'$.
Therefore $v \notin S'$ and $x \notin S'$.

There are three ways how solution $S'$ can interact with $M'$ and $X'$ that we need to address.

\begin{figure}[!h]
\begin{center}
\begin{tikzpicture}[every node/.style={draw, fill, circle, inner sep=1.5pt}, scale=0.99904]

% left figure
\begin{scope}
\node[label=above:$v$, red] (v) at (6.5,1.5) {};

\node[setXvertex, label={above:$x$}] (x1) at (1,0) {}; 
\draw[very thick, bend right=20, red] (v) to (x1);
\foreach \i in {2,...,6}{
  \node[setXvertex] (x\i) at (\i,0) {}; 
  \draw[thick, bend right=20, vertexVorangeEdge] (v) to (x\i);
 }
\foreach \i in {7,...,12}{
  \node[setXvertex] (x\i) at (\i,0) {}; 
  \draw[thick, bend left=20, vertexVorangeEdge] (v) to (x\i);
 }

\begin{scope}[on background layer]
    \node[rectangle,color=setX1BoxDraw,draw,fill=setX1BoxFill,dashed,inner sep=0.5cm,rounded corners,fit=(x1)(x12),label={right:$X'$}] (XPrimeFitter) {};
\end{scope}

% expansion Q sets for mhat1 and mhat2
\begin{scope}[on background layer]
    \node[rectangle,color=setM1BoxDraw,draw,fill=setM1BoxFill,dashed,inner sep=0.3cm,rounded corners,fit=(x1)(x4),label={200:$Q_{\hat{m}_1}$}] (Qmhat1fitter) {};
    \node[rectangle,color=setM1BoxDraw,draw,fill=setM1BoxFill,dashed,inner sep=0.3cm,rounded corners,fit=(x5)(x8),label={345:$Q_{\hat{m}_2}$}] (Qmhat2fitter) {};

\end{scope}

% M' vertices
\node[label={below:$\hat{m}_1$}] (mhat1) at (4.5, -1.5) {};
\node[label={below:$\hat{m}_2$}] (mhat2) at (6.5, -1.5) {};
\node[] (mprime) at (8.5, -1.5) {};
  
\begin{scope}[on background layer]
    \node[rectangle,color=setM1BoxDraw,draw,fill=setM1BoxFill,dashed,inner sep=0.5cm,rounded corners,fit=(mhat1)(mprime),label={350:$M'$}] (MPrimeFitter) {};
    \node[rectangle,color=setMprimeBoxDraw,draw,fill=setMprimeBoxFill,dashed,inner sep=0.3cm,rounded corners,fit=(mhat1)(mhat2),label={right:$\hat{M}$}] (MHatFitter) {};
\end{scope}

% expansion edges
\foreach \i in {1,...,4} {
    \draw[very thick, bend right=20, setXprimeEdge] (x\i) to (mhat1);
}
\foreach \i in {5,6} {
    \draw[very thick, bend right=20, setXprimeEdge] (x\i) to (mhat2);
}
\foreach \i in {7,8} {
    \draw[very thick, bend left=20, setXprimeEdge] (x\i) to (mhat2);
}
\foreach \i in {9,...,12} {
    \draw[very thick, bend left=20, setXprimeEdge] (x\i) to (mprime);
}

% draw the solution S'
\begin{scope}[on background layer]
\foreach \i in {x2,x3,x5,x6,x7,mprime} {
\filldraw[blue] (\i) circle (0.2);
}
\draw[dashed, thick, bend left=20, blue, ->] (x2) to (mhat1);
\draw[dashed, thick, bend right=20, blue, ->] (x3) to (v);
\draw[dashed, thick, bend left=20, blue, ->] (x5) to (mhat2);

\end{scope}

\end{scope}

\draw[line width=2pt,-{Stealth[round]}] (6.5,-2.2) to (6.5,-2.8);

% right figure
\begin{scope}[shift={(0,-4.8)}]
\node[label=above:$v$, red] (v) at (6.5,1.5) {};

\node[setXvertex, label={above:$x$}] (x1) at (1,0) {}; 
\draw[very thick, bend right=20, red] (v) to (x1);
\foreach \i in {2,...,6}{
  \node[setXvertex] (x\i) at (\i,0) {}; 
  \draw[thick, bend right=20, vertexVorangeEdge] (v) to (x\i);
 }
\foreach \i in {7,...,12}{
  \node[setXvertex] (x\i) at (\i,0) {}; 
  \draw[thick, bend left=20, vertexVorangeEdge] (v) to (x\i);
 }

\begin{scope}[on background layer]
    \node[rectangle,color=setX1BoxDraw,draw,fill=setX1BoxFill,dashed,inner sep=0.5cm,rounded corners,fit=(x1)(x12),label={right:$X'$}] (XPrimeFitter) {};
\end{scope}

% expansion Q sets for mhat1 and mhat2
\begin{scope}[on background layer]
    \node[rectangle,color=setM1BoxDraw,draw,fill=setM1BoxFill,dashed,inner sep=0.3cm,rounded corners,fit=(x1)(x4),label={200:$Q_{\hat{m}_1}$}] (Qmhat1fitter) {};
    \node[rectangle,color=setM1BoxDraw,draw,fill=setM1BoxFill,dashed,inner sep=0.3cm,rounded corners,fit=(x5)(x8),label={345:$Q_{\hat{m}_2}$}] (Qmhat2fitter) {};

\end{scope}

% M' vertices
\node[label={below:$\hat{m}_1$}] (mhat1) at (4.5, -1.5) {};
\node[label={below:$\hat{m}_2$}] (mhat2) at (6.5, -1.5) {};
\node[] (mprime) at (8.5, -1.5) {};
  
\begin{scope}[on background layer]
    \node[rectangle,color=setM1BoxDraw,draw,fill=setM1BoxFill,dashed,inner sep=0.5cm,rounded corners,fit=(mhat1)(mprime),label={350:$M'$}] (MPrimeFitter) {};
    \node[rectangle,color=setMprimeBoxDraw,draw,fill=setMprimeBoxFill,dashed,inner sep=0.3cm,rounded corners,fit=(mhat1)(mhat2),label={right:$\hat{M}$}] (MHatFitter) {};
\end{scope}

% expansion edges
\foreach \i in {1,...,4} {
    \draw[very thick, bend right=20, setXprimeEdge] (x\i) to (mhat1);
}
\foreach \i in {5,6} {
    \draw[very thick, bend right=20, setXprimeEdge] (x\i) to (mhat2);
}
\foreach \i in {7,8} {
    \draw[very thick, bend left=20, setXprimeEdge] (x\i) to (mhat2);
}
\foreach \i in {9,...,12} {
    \draw[very thick, bend left=20, setXprimeEdge] (x\i) to (mprime);
}

% draw the solution S
\begin{scope}[on background layer]
\foreach \i in {v,mhat1,mhat2,mprime} {
\filldraw[blue] (\i) circle (0.2);
}

\end{scope}

\end{scope}

\end{tikzpicture}
\end{center}
\caption{Illustration of the first part of the proof of correctness of \autoref{rule:expansion}.}
\label{fig:setting4}
\end{figure}

Firstly, suppose that $M' \nsubseteq S'$ and let $\widehat{M} = M' \setminus S'$.
See \autoref{fig:setting4} for an illustration.
We have that for each vertex $\widehat{m} \in \widehat{M}$ it must be that $|Q_{\widehat{m}} \cap S'| \geq 2$. Indeed, if this is not the case, there would be a $P_d$ in $G' \setminus S'$ which uses the $(d-2)$ vertices of $Q_{\widehat{m}}$ not in $S'$ and the vertices $\widehat{m}$ and $v$.
Consider a set $S = (S' \cup \widehat{M} \cup \{v\}) \setminus \bigcup_{\widehat{m} \in \widehat{M}} Q_{\widehat{m}}$. Observe, that any $P_d$ which uses some vertex from $X'$ must contain at least one of the vertices in $M'$ or $v$, because $N(X') \subseteq M' \cup \{v\}$. Therefore the set $S$ is a solution for the reduced graph~$G'$ as it contains both $M'$ and $v$. The set $S$ is also a solution for $G$ as it contains $v$ and therefore covers any $P_d$ which might use the deleted edge $\{x, v\}$. Finally, $|S| \leq |S'|$, because for each $\widehat{m}$ that we add into $S$, we remove at least two vertices of $Q_{\widehat{m}}$ from $S$ and therefore we have that $|\widehat{M} \cup \{v\}| \leq 2|\widehat{M}|\le |\bigcup_{\widehat{m} \in \widehat{M}} Q_{\widehat{m}}|$.

Secondly, assume that $M' \subseteq S'$ and there is some $x' \in X'$ such that $x' \in S'$. We construct the set $S = (S' \setminus X') \cup \{v\}$. Again, observe that $S'$ is a solution for $G'$ because any $P_d$ which uses some vertex from $X'$ must contain at least one of the vertices in $M'$ or $v$ and both are fully contained in $S$. We also have that $S$ is a solution for $G$, again, as it contains $v$ and therefore covers any $P_d$ which might use the deleted edge $\{x, v\}$. Finally, $|S| \leq |S'|$, because we have the assumption that there is some $x' \in X'$ and $x' \in S'$.

Lastly, assume that $M' \subseteq S'$ and $X' \cap S' = \emptyset$.
% See \autoref{fig:setting5} for an illustration.
In this case, the $P_d$ that we found in $G \setminus S'$ must be of the form $P = (x, v, u_3,\ldots, u_d)$ and $u_3,\ldots, u_d \notin (M' \cup X')$.
The existence of such $P_d$ also gives us that $v, u_3,\ldots, u_d \notin S'$.
Let $x'$ be an arbitrary vertex of $X' \setminus \{x\}$ (note that $|X'| \geq (d-1)|M'| \geq 3$).
Then the $P_d$ of the form $P' = (x', v, u_3, \ldots, u_d)$ can be found in both $G$ and $G'$, which contradicts the fact that $S'$ is a solution in $G'$.

% \begin{figure}[!h]
% \begin{center}
% \input{src/tikz_images/tikz_image_06_high_degree_rule_correctness_part_03.tex}
% \end{center}
% \caption{Illustration of the last part of the proof of correctness of \autoref{rule:expansion}.}
% \label{fig:setting5}
% \end{figure}

To sum up, we have shown that when we delete the edge $\{x, v\}$ from $G$, then for any solution $S'$ for $G'$ which is not also a solution for $G$, we can always find a new solution $S$, $|S| \leq |S'|$ which is a solution for both $G'$ and~$G$.
\end{proof}
As the application of the rule only requires finding a largest matching adjacent to $v$, classifying the vertices of $N(v)$, and finding a $(d-1)$-expansion and these tasks can be done in polynomial time, the rule can be applied in polynomial time.

\section{4-PVC Kernel with Quadratic Number of Edges}\label{sec:4PVC}
\toappendix{\section{Additional Material to \autoref{sec:4PVC}}}%
Let $(G=(V,E),k)$ be an instance reduced by exhaustively employing Reduction Rules~\ref{rule:component}--\ref{rule:expansion}.
Then the maximum degree in $G$ is at most $(d+2)(k+1)=6k + 6$.
Furthermore, assume that the algorithm of \autoref{prop:greedy_packing} actually returned an inclusion-wise maximal packing $\PP$ in $G$ with at most $k$ $4$-paths instead of answering immediately.
Let $P=V(\PP)$, and let $A = V \setminus P$.
There are at most $4k$ vertices in $P$.
Each connected component in $G[A]$ is a 4-path free graph, otherwise we would be able to increase the size of the packing $\PP$.
Since the instance is reduced with respect to \autoref{rule:component},  each connected component in $G[A]$ is connected to $P$ by at least one edge.

To show that an instance reduced with respect to all the above rules has a quadratic number of edges, it suffices to count separately the number of edges incident on $P$ and the number of edges in $G[A]$.
Since the maximum degree in $G$ is at most $6k + 6$ and there are at most $4k$ vertices in $P$, there are at most $4k \cdot (6k + 6) = 24k^2 + 24k$ edges incident on $P$.

To count the edges in $G[A]$, we first observe that a connected 4-path free graph is either a triangle, or a star (possibly degenerate, i.e., with at most 3 vertices).
Here a~\textit{$q$-star} is a~graph with vertices $\{c, l_1,\ldots,l_q\}$, $q \ge 0$ and edges $\{ \{c,l_i\} \mid i \in \{1,\ldots,q\} \}$. Vertex~$c$ is called \emph{a~center}, vertices $\{l_1,\ldots,l_q\}$ are called \emph{leaves}. The term \textit{star} will be used for a $q$-star with an arbitrary number of leaves.
Note that, a graph with a single vertex is a $0$-star, a graph with two vertices and a single edge is a $1$-star, and a $3$-path is a $2$-star. A~\textit{triangle} is a~cycle on three vertices.

Secondly, as the instance is reduced with respect to \autoref{rule:component}, each connected component in $G[A]$ is connected to $P$ by at least one edge. Therefore, there are at most $24k^2 + 24k$ connected components in $G[A]$, as there are only that many edges going from $P$ to $A$. Next, we provide an observation about stars in $G[A]$.

\begin{observation}[\appmark]
\label{obs:q_star_size}
For each $q$-star $(c, l_1, l_2, \ldots, l_q)$ in $G[A]$, there are at most two vertices in the $q$-star which are not connected to $P$ by any edge in $G$, one possibly being the center $c$ and the other possibly being a leaf $l_i$ of the star.
\end{observation}
\toappendix{%
\begin{proof}[Proof of \autoref{obs:q_star_size}]
% \begin{proof}
Suppose that there is a $q$-star $(c, l_1, l_2, \ldots, l_q)$ in $G[A]$ with two leaves $l_i, l_j$ not being connected to $P$ in $G$.
This means that $N(l_i) = N(l_j) = \{c\}$ which contradicts the fact that $G$ is reduced with respect to \autoref{rule:degree-one}.
\end{proof}
}%toappendix

%We are now ready to count the edges in $G[A]$.

\begin{observation}[\appmark]
\label{obs:4PVC_edges_in_A}
There are at most $72k^2 + 72k$ edges in $G[A]$.
\end{observation}

\toappendix{%
% \begin{proof}
\begin{proof}[Proof of \autoref{obs:4PVC_edges_in_A}]
Let $\mathcal{T}$ be the collection of triangles in $G[A]$ and let $\mathcal{S}$ be the collection of stars in~$G[A]$.
Let $t$ be the number of edges between $P$ and $\mathcal{T}$ and $s$ be the number of edges between $P$ and $\mathcal{S}$.
We know that $s+t \le 24k^2+24k$, $|\mathcal{T}| \le t$, and $|\mathcal{S}| \le s$.
There are at most $3t$ edges and vertices in $\mathcal{T}$ as each triangle has three edges and three vertices.
There are at most $s+2|\mathcal{S}|\le 3s$ vertices and edges in $\mathcal{S}$ as there are at most 2 vertices in each star that are not incident to any edge going from $P$.
Therefore there are at most $3s+3t=3(s+t) \le 3(24k^2 + 24k) = 72k^2 + 72k$ edges in $G[A]$ in total.
\end{proof}
}%toappendix

We conclude this section with the final statement about our kernel.

\begin{theorem}[\appmark]
\label{thm:4PVC_kernel}
$4$-\textsc{Path Vertex Cover} admits a kernel with $96k^2 + 96k$ edges, where $k$ is the size of the solution.
\end{theorem}
\toappendix{%
\begin{proof}[Proof of \autoref{thm:4PVC_kernel}]
The kernelization algorithm simply applies Reduction Rules~\ref{rule:component}--\ref{rule:expansion} exhaustively.
Since all the reduction rules are applicable in polynomial time, the kernelization algorithm runs in polynomial time.

Let $(G=(V,E),k)$ be an instance reduced with respect to Reduction Rules~\ref{rule:component}--\ref{rule:expansion}.
From the previous observations, there are at most $24k^2 + 24k$ edges incident on $P$ and there are at most $72k^2 + 72k$ edges in $G[A]$, which are all the edges in $G$. Therefore the total count of edges in $G$ is at most $96k^2 + 96k$.
\end{proof}
}%toappendix

\section{5-\PVC{} Kernel with Quadratic Number of Edges}\label{sec:5PVC}
\toappendix{\section{Additional Material to \autoref{sec:5PVC}}}%
The idea is completely analogous to the previous section.
We employ the following characterization.

A~\textit{star with a triangle} is formed by connecting two leaves of a star with an edge. %$q$-star, $q \geq 2$ with an edge.
A~\textit{bi-star} is formed by connecting the centers of two stars with an edge. %a $q_1$-star and a $q_2$-star with an edge.
% Note that a $q$-star is a bi-star formed by connecting $(q-1)$-star with a $0$-star.

\begin{lemma}[{Červený and Suchý~\cite[Lemma 4]{CervenyS19}}] \label{lem:5PVC_structure}
A connected 5-path free graph is either a graph on at most 4 vertices, a star with a triangle, or a bi-star.
\end{lemma}

\ifappendix{We defer the rest of the proof to appendix.}
\toappendix{%
Let $(G=(V,E),k)$ be an instance reduced by exhaustively employing Reduction Rules \ref{rule:component}--\ref{rule:expansion}.
Then the maximum degree in $G$ is at most $(d+2)(k+1)=7k + 7$.
Furthermore, assume that the \autoref{algorithm:greedy_packing} actually returned an inclusion-wise maximal packing $\PP$ in $G$ with at most $k$ $5$-paths instead of answering immediately.
Let $P=V(\PP)$, and let $A = V \setminus P$.
There are at most $5k$ vertices in $P$.
Each connected component in $G[A]$ is a $5$-path free graph, otherwise we would be able to increase the size of the packing $\PP$.
Since the instance is reduced with respect to \autoref{rule:component},  each connected component in $G[A]$ is connected to $P$ by at least one edge.

To show that an instance reduced with respect to all the above rules has quadratic number of edges, it suffices to count separately the number of edges incident on $P$ and the number of edges in $G[A]$.
Since the maximum degree in $G$ is at most $7k + 7$ and there are at most $5k$ vertices in $P$, there are at most $5k \cdot (7k + 7) = 35k^2 + 35k$ edges incident on $P$.

% To count the edges in $G[A]$, we first need to argue about the structure of $5$-path free graphs.
To count the edges in $G[A]$, we use \autoref{lem:5PVC_structure} about the structure of $5$-path free graphs.
% % \begin{definition}\label{def:starwithtriangle}
% A~\textit{star with a triangle} is formed by connecting two leaves of a $q$-star, $q \geq 2$ with an edge.
% % \end{definition}
% %
% % \begin{definition}\label{def:distar}
% A~\textit{bi-star} is formed by connecting the centers of a $q_1$-star and a $q_2$-star with an edge.
% % \end{definition}
% Note that a $q$-star is a bi-star formed by connecting $(q-1)$-star with a $0$-star.
%
% \begin{lemma}[{Červený and Suchý~\cite[Lemma 4]{CervenyS19}}]
% A connected 5-path free graph is either a graph on at most 4 vertices, a star with a triangle, or a bi-star.
% \end{lemma}

Secondly, as the instance is reduced with respect to \autoref{rule:component}, each connected component in $G[A]$ is connected to $P$ by at least one edge. Therefore, there are at most $35k^2 + 35k$ connected components in $G[A]$, as there are only that many edges going from $P$ to $A$. Next, we provide an observation about stars with a triangle and bi-stars in $G[A]$.

\begin{observation}
For each star with a triangle in $G[A]$, there are at most $4$ vertices in the star with triangle which are not connected to $P$ by any edge in $G$, three of them possibly being the vertices of the triangle and the other possibly being a leaf of the star.
\end{observation}
\begin{proof}
Suppose that there is a star with a triangle in $G[A]$ with two leaves $l_i, l_j$ (not involved in the triangle) not being connected to $P$ in $G$.
This means that $N(l_i) = N(l_j) = \{c\}$ for $c$ the center of the star, which contradicts the fact that $G$ is reduced with respect to \autoref{rule:degree-one}.
\end{proof}

\begin{observation}
For each bi-star in $G[A]$, there are at most $4$ vertices in the bi-star which are not connected to $P$ by any edge in $G$, the two centers and one leaf for each center.
\end{observation}
\begin{proof}
Suppose that the two centers of the bi-star are $c$ and $c'$ with leaves $l_1, l_2, \ldots, l_{q_1}$ adjacent to $c$ and leaves $l'_1, l'_2, \ldots, l'_{q_1}$ adjacent to $c'$.
Suppose that two leaves $l_i, l_j$ are not being connected to $P$ in $G$.
This means that $N(l_i) = N(l_j) = \{c\}$, which contradicts the fact that $G$ is reduced with respect to \autoref{rule:degree-one}.
Similarly, if two leaves $l'_i, l'_j$ are not being connected to $P$ in $G$, then $N(l'_i) = N(l'_j) = \{c'\}$ contradicting the fact that $G$ is reduced with respect to \autoref{rule:degree-one}.
\end{proof}

We are now ready to count the edges in $G[A]$.

\begin{observation}
There are at most $210k^2 + 210k$ edges in $G[A]$.
\end{observation}
\begin{proof}
Let $\mathcal{R}$ be the collection of connected components of $G[A]$ that are of size at most $4$,
let $\mathcal{T}$ be the collection of stars with triangles in $G[A]$, and let $\mathcal{B}$ be the collection of bi-stars in $G[A]$.
Let $r$ be the number of edges between $P$ and $\mathcal{R}$, $t$ be the number of edges between $P$
and $\mathcal{T}$ and $b$ be the number of edges between $P$ and $\mathcal{B}$.
We have that $|\mathcal{R}|\le r$, $|\mathcal{T}|\le t$, $|\mathcal{B}|\le b$, and $r+t+b \le 35k^2+35k$.
There are at most $6|\mathcal{R}| \le 6r$ edges in $\mathcal{R}$ as each connected component in $\mathcal{R}$ has at most $4$ vertices
and, hence, at most $6$ edges.
A star with a triangle with $p$ edges from $P$ has at most $p+4$ vertices and, hence, at most $p+4$ edges.
Therefore, there are at most $t+4|\mathcal{T}| \le 5t$ edges in~$\mathcal{T}$.
A bi-star with $p$ edges from $P$ has at most $p+4$ vertices and, hence, at most $p+3$ edges.
Thus, we have at most $b+3|\mathcal{B}| \le 4b$ edges in~$\mathcal{B}$.
Consequently, there are at most $6r+5t+4b \le 6(r+t+b) \le 6(35k^2+35k)= 210k^2+210k$ edges in $G[A]$ in total.
\end{proof}
}%to appendix

We conclude this section with the final statement about our kernel.

\begin{theorem}[\appmark]
\label{thm:5PVC_kernel}
$5$-\textsc{Path Vertex Cover} admits a kernel with $245k^2 + 245k$ edges, where $k$ is the size of the solution.
\end{theorem}

\toappendix{%
\begin{proof}[Proof of \autoref{thm:5PVC_kernel}]
The kernelization algorithm simply applies Reduction Rules~\ref{rule:component}--\ref{rule:expansion} exhaustively.
Since all the reduction rules are applicable in polynomial time, the kernelization algorithm runs in polynomial time.

Let $(G=(V,E),k)$ be an instance reduced with respect to Reduction Rules~\ref{rule:component}--\ref{rule:expansion}.
From the previous observations, there are at most $35k^2 + 35k$ edges incident on $P$, and there are at most $210k^2 + 210k$ edges in $G[A]$, which are all the edges in $G$. Therefore the total count of edges in $G$ is at most $245k^2 + 245k$.
\end{proof}
}%to appendix

\section{\textsc{$d$-PVC} Kernel with $\bigo(k^4d^{2d+9})$ Edges}
\label{sec:dpvc_general_kernel}
\toappendix{\section{Additional Material to \autoref{sec:dpvc_general_kernel}}}%

In this section, we give a kernelization algorithm for \textsc{$d$-PVC} with $d \geq 6$.

\textit{An intuition behind the approach.}
The kernelization algorithm marks some vertices and edges, which it wants to keep, and throws away the rest. Essentially, the kernelization creates a subgraph $\widehat{G}$ of the input graph $G$. For correctness of the algorithm, we want to show that if there is a $d$-path $P$ in $G$ which misses some set of vertices $S$ (a prospective solution), then we will also find some $d$-path $P'$ in $\widehat{G}$, which also misses the set $S$.

We begin by finding a maximal packing $\M$ in $G$ and we keep in $\widehat{G}$ all vertices $M$ of the packing and all edges between them. Now, on one hand, if the path $P$ would be completely contained in $M$, then trivially the path appears also in $\widehat{G}$. On the other hand, the path $P$ cannot be completely outside of $M$. Thus, the path $P$ crosses between $M$ and outside of $M$ at least once.
% We can characterize this crossing as some desire to connect two vertices of $M$ through the outside of $M$ with a path of some desired length.
This corresponds to vertices of $M$ being connected by a path of prescribed length outside of $M$.
We later formalize this as a ``request''.

To get more structure, we leverage the behavior of DFS trees of the connected components outside of $M$.
With the DFS trees we identify vertices, which are ``crucial'' for the requests, and we further split the requests into ``sub-requests'' according to the ``crucial'' vertices.

The algorithm is inspired by Dell and Marx~\cite{DellM18}.
However, while the considered problems have similarities, many ideas are not translatable.
In particular, they could afford to consider all ``sub-requests'' and keep $\Omega(k)$ vertices for each without affecting their bound (cf.~\cite[p. 23]{DellM18}).
We had to be more careful in which ``sub-requests'' we consider and we need to employ \autoref{lemma:dpvc_sol_mod} (below) to only keep $d^{\bigo(d)}$ vertices and edges for each such ``sub-request'' to achieve our precise bound.
Also, to achieve the edge bound, we need to keep track of the purpose for which the individual vertices were marked, which makes it hard to split the algorithm into small self-contained steps.

\textit{Formal definitions.}
More formally, assume that we are given an instance $(G=(V,E), k)$ of \textsc{$d$-PVC}.
We start by running the algorithm of \autoref{prop:greedy_packing} on the instance.
If it answers directly, then we are done.
Otherwise it returns a maximal packing $\M$ in $G$.
Let $M$ be the vertices of the packing $\M$. Recall that $|M| \leq dk$.

Let $G' = G \setminus M$, i.e., $G'$ is the graph outside the packing $\M$. Label the connected components of $G'$ as $G'_1, G'_2, \dots, G'_t$. For each component $G'_i$ pick an arbitrary vertex $r_i \in G'_i$ and compute a \emph{depth-first search} tree $T_i$ of $G'_i$ rooted at $r_i$.
Note that $V(T_i) = V(G'_i)$.
Let $\F$ denote the forest consisting of all the trees $T_i, i \in [t]$.
Note that $V(\F) = V(G')$.

For a rooted forest $F$ and its vertex $v \in V(F)$, $\sub(v)$ denotes the set of vertices of a maximal subtree of $F$ rooted at $v$ and $\anc(v)$ the set of ancestors of $v$ in $F$, i.e., $\anc(v) = \{u \mid u \in V(F), v \in \sub(u)\}$. Note that $v \in \anc(v)$.
% The \emph{depth} of a rooted tree is the length of the longest path between the root of the tree and one of the leaves of the tree. The \emph{depth} of a rooted forest is the maximum depth among its trees.

We provide the following observations regarding the DFS trees $T_i$ and the forest $\F$.

\begin{observation}[\appmark]
\label{observation:trees}
\begin{enumerate}[(a)]
\item Each $y \in V(\F)$ has at most $d-1$ ancestors.
%The depth of each tree $T_i$ is at most $d-2$. In particular, the depth of $\F$ is at most $d-2$.
\label{observation:trees_depth}
\item For two vertices $u,v \in V(\F)$ which are not in ancestor-descendant relation, we have $\sub(u) \cap \sub(v) = \emptyset$ and $\{u,v\} \notin E(G)$. \label{observation:trees_disjoint}
\item If $C$ is a connected subgraph of $G'$, then there is a vertex $w \in V(C)$ such that $V(C) \subseteq \sub(w)$.
\label{observation:connected_subgraph}
% \item For two vertices $u,v \in V(\F)$ which belong to the same tree $T_i$ in $\F$ and have $\sub(u) \cap \sub(v) = \emptyset$, we have that every path in the component $G'_i$, which uses some vertex from $\sub(u)$ and some vertex from $\sub(v)$, must also use some common ancestor $w \in \anc(u) \cap \anc(v)$ of $u$ and $v$.
\end{enumerate}
\end{observation}

\toappendix{%
\begin{proof}[Proof of \autoref{observation:trees}]
\begin{enumerate}[(a)]
\item %Assume that the depth of some tree $T_i$ is at least $d-1$. Then there is a $d$-path in $T_i$ going from the root $T_i$ to some leaf of $T_i$.
Assume that some $y \in V(\F)$ has at least $d$ ancestors. Then these ancestors form a $d$-path between $y$ and the root of the corresponding tree $T_i$.
Since $V(T_i) \cap M = \emptyset$, it contradicts the maximality of packing $\M$.
\item Directly follows from the fact that $\F$ is a rooted forest and each $T_i$ is a depth-first search tree.
\item Since $C$ is connected, there is $i \in [t]$ such that $C$ is a subgraph of $G'_i$.
Let $w$ be a vertex in $V(C)$ such that $\anc(w) \cap V(C) = \{w\}$.
If $V(C) \subseteq \sub(w)$, then we are done.
Otherwise, let $x \in V(C) \setminus \sub(w)$ and let $P_x$ be a path from $w$ to $x$ in $C$.
Let $x'$ be the first vertex outside $\sub(w)$ on that path, and $y$ be the previous vertex to $x'$ on $P_x$.
Then, by (\ref{observation:trees_disjoint}), either $y$ is an ancestor of $x'$ or vice versa.
In the first case we have $x'$ in $\sub(w)$ contradicting the choice of $x'$.
In the second case we have $x'$ in $V(C)$. Moreover, as it is an ancestor of $y \in \sub(w)$ and $x'$ is not in $\sub(w)$, it is an ancestor of $w$ in $V(C)$, contradicting the choice of $w$.
This shows that $V(C) \subseteq \sub(w)$. \qedhere
\end{enumerate}
\end{proof}
}%to appendix

A triple $(f,l,X)$ is \emph{an $X$-request}, if $l \in \{1, \dots, d-1\}$, $X \subseteq V(G)$, and either $f=\{u,v\} \subseteq X$, $u \neq v$, or $f=\{x\} \subseteq X$.
For an $X$-request $(f,l,X)$ and $H \subseteq V(G)$, we use $\mathcal{P}^{H}_{f,l}$ to denote the set of paths $P$ of length $l$ in $G[H \cup f]$ such that each $x \in f$ is an endpoint of $P$.
In particular, if $f=\{u,v\}$, $u \neq v$, then $u$ and $v$ are the endpoints of $P$ and if $f=\{x\}$, then one endpoint of $P$ is $x$ and the other can be any vertex of $H$.
A set $H \subseteq V(G)$ is said to \emph{satisfy an $X$-request} $(f,l,X)$ if $\mathcal{P}^{H}_{f,l} \neq \emptyset$.

An $M$-request $(f,l,M)$ will be simply denoted as $\req{(f,l)}$ and called \emph{a request}.

In the next paragraphs we cover the notion of the ``crucial'' vertices mentioned in the earlier intuition.
Roughly speaking, a vertex of $F$ is ``crucial'' if the set of its descendants (vertices of its subtree) satisfies some request.

%\begin{observation}\label{observation:number_of_requests}
%There are at most $(\binom{dk}{2} + dk)(d-1) = \bigo(k^2d^3)$ possible requests.
%\end{observation}
%\begin{proof}
%The size of $|M|$ is at most $dk$. Hence, the number of possibilities to pick the set $f$ is $\binom{dk}{2}$ when $|f| = 2$ and $dk$ when $|f| = 1$. For each unique set $f$ we have $d-1$ different choices for $l$.
%\end{proof}

For each request $\req{(f, l)}$ we define the set $Y_{f,l} \subseteq V(\F)$ as the set of all vertices $v \in V(\F)$ such that $\sub(v)$ satisfies the request $\req{(f,l)}$.
Note that if $v \in Y_{f,l}$, then also $w \in Y_{f,l}$ for every $w \in \anc(v)$, since $\sub(w) \supseteq \sub(v)$.
Thus, if $Y_{f,l} \cap V(G'_i) \neq \emptyset$ for some $i$, then in particular $r_i \in Y_{f,l}$ and $Y_{f,l}$ induces a connected subtree of $T_i$.
The request $\req{(f,l)}$ is \emph{resolved} if the subforest $\F[Y_{f,l}]$ has at least $k+d+1$ leaves.

Recall, that in the intuition we examined some $d$-path $P$ in $G$. The resolved request basically ensures that there are at least $k+d+1$ disjoint paths in $G$ which satisfy said request. The idea is that the prospective solution may compromise at most $k$ of these paths and the other parts of $P$ may compromise at most $d$ of these paths. As we will keep exactly $k+d+1$ of these  disjoint paths in $\widehat{G}$, we can be sure, that at least one of them will always be usable to reroute some part of $P$ which will help us to find the desired path $P'$ in $\widehat{G}$.

Now, we focus on the unresolved requests.
%The set $\Y$ denotes the union of all sets $Y_{f,l}$ corresponding to requests $\req{(f,l)}$ which are \emph{not} resolved, i.e., let $\R^*$ be the set of all requests $\req{(f,l)}$ which are \emph{not} resolved, then $\Y = \bigcup_{\req{(f,l)} \in \R^*} Y_{f,l}$.
Let $\R^*$ be the set of all requests $\req{(f,l)}$ which are \emph{not} resolved and let $\Y = \bigcup_{\req{(f,l)} \in \R^*} Y_{f,l}$.
%Label the subtrees of $\F[\Y]$ as $T^\Y_1, T^\Y_2, \dots, T^\Y_p$.

%\begin{observation}
%The subforest $\F[\Y]$ has at most $\bigo(k^3d^4)$ leaves.
%\end{observation}
%\begin{proof}
%By \autoref{observation:number_of_requests}, there are at most $\bigo(k^2d^3)$ possible requests and for each such request we add at most $k+d+1$ leaves to $F[\Y]$ as we include only those sets $Y_{f,l}$ corresponding to unresolved requests which, by definition, have at most $k+d+1$ leaves in $\F[Y_{f,l}]$. Therefore, the maximum number of leaves in $\F[\Y]$ is $\bigo((k^2d^3)\cdot(k+d+1)) = \bigo(k^3d^4)$.
%\end{proof}

%Let $\C$ be the set of vertices $v \in V(\F) \setminus \Y$ such that $\anc(v) \cap \Y \neq \emptyset$. Let $\D = V(\F) \setminus (\Y \cup \C)$. Label the subtress of $\F[\C]$ as $T^\C_1, T^\C_2, \dots, T^\C_q$ and the subtrees of $\F[\D]$ as $T^\D_1, T^\D_2, \dots, T^\D_r$. Note that $V(F) = \Y \cup \C \cup \D$.

%Let $C_1, C_2, \dots, C_q$ be the vertex sets of connected components of $\F \setminus \Y$. Let $\C = \bigcup_{i \in [q]} C_i$.

% \begin{figure}[!h]
% \begin{center}
% \includegraphics[width=\textwidth]{src/images/dpvc_kernel_sketch_01.pdf}
% \end{center}
% \caption{An overview of the notation.}
% \label{fig:dpvc_kernel_notation}
% \end{figure}

We are now getting to the notion of sub-requests. These have either one endpoint in $M$ and the other in $\Y$, or both endpoints in $\Y$, or only one prescribed endpoint, which is in $\Y$. Note that if the two endpoint are in $\Y$, then, by \autoref{observation:trees}, either one of the endpoints is an ancestor of the other, or there is no path connecting them outside $(M \cup \Y)$.

An $(M \cup \Y)$-request $(g,j,M \cup \Y)$ will be simply denoted as $\subreq{(g,j)}$ and called \emph{a sub-request} if there exists $y \in \Y$ such that $y \in g$ and $g \subseteq M \cup \anc(y)$.
In particular, either $g=\{y\}$ or one of the vertices in $g$ is $y$ and the other vertex is in $M \cup \anc(y)$.

Even though we will not formally define a \textit{resolved sub-request}, later we will actually show that the sub-request is ``resolved'' if there are at least $2d$ paths satisfying it.

\begin{algorithm}[h!]
\caption{Marking procedure \textbf{Mark}}
\label{alg:mark}
\BlankLine
Let $M$, $\F$, and $\Y$ be as in the text. \\
Mark all the vertices and edges in $G[M \cup \Y]$. \label{alg:mark_my}\\
\ForEach{resolved request $\req{(f,l)}$\label{alg:mark_foreach_resolved_start}}{
    Pick arbitrary $k+d+1$ leaves $h_1, h_2, \dots, h_{k+d+1}$ of $\F[Y_{f,l}]$.\\
    \ForEach{leaf $h_i$}{
        Pick an arbitrary path $P$ from $\mathcal{P}^{\sub(h_i)}_{f,l}$.\\
        Mark the vertices and edges of $P$.\label{alg:mark_foreach_resolved_end}\\
    }
}
\ForEach{$y \in \Y$\label{alg:mark_foreach_y_start}}{
    \ForEach{sub-request $\subreq{(g,j)}$ such that $g \subseteq M \cup \anc(y)$ and $g \not\subseteq M$}{
        Let $C_1, C_2, \dots, C_{q'}$ be the vertex sets of the connected components of $G \setminus (M \cup \Y)$ such that for all $i \in [q']$, $N(C_i) \cap \Y \subseteq \anc(y)$ and $C_i$ satisfies $\subreq{(g,j)}$.\label{alg:components}\\
    \If{$q' \geq 2d$}{\label{alg:mark_c_bound_start}
        \ForEach{$i \in [2d]$}{
            Pick an arbitrary path $P \in \PP^{C_i}_{g,j}$.\\
            Mark the vertices and edges of $P$.\\
        }
    }\label{alg:mark_c_bound_end}
    \Else{\label{alg:mark_mark2_start}
        \ForEach{$i \in [q']$}{
            Run the marking procedure \textbf{Mark~2} on $\Big(\subreq{(g,j)},C_i,\emptyset\Big)$. \label{alg:mark_foreach_y_end}\label{alg:mark_mark2_end}\\
        }
    }
    }
}

\end{algorithm}

\begin{algorithm}[h!]
\caption{Marking procedure \textbf{Mark 2}$\Big( \subreq{(g,j)}, C_i, W \Big)$}
\label{alg:mark2}
\BlankLine
\If{$|W| \leq 2d$ and $\PP^{C_i \setminus W}_{g,j} \neq \emptyset$}{
    Let $P \in \PP^{C_i \setminus W}_{g,j}$.\\
    Mark all the edges and vertices of $P$.\\
        \ForEach{$v \in V(P) \setminus g$}{
            $W' = W \cup \{v\}$\\
            Call \textbf{Mark 2} on $\Big(\subreq{(g,j)},C_i,W'\Big)$
        }
    }
\end{algorithm}

\textit{Description of the algorithm.} We are now ready to describe the kernelization algorithm. As we mentioned earlier, the algorithm first marks some vertices and edges, which it wants to keep, and it deletes the rest of the graph. Therefore, the core of the algorithm is the marking procedure. In our case, the main procedure is called \textbf{Mark} which in turn uses a procedure called \textbf{Mark~2}. These procedures are described in \autoref{alg:mark} and \autoref{alg:mark2}, respectively.

Let us now give an insight into how the procedures \textbf{Mark} and \textbf{Mark~2} were constructed. We will start with \textbf{Mark}.

The lines \ref{alg:mark_foreach_resolved_start}--\ref{alg:mark_foreach_resolved_end} deal with the resolved requests.
Essentially, by preserving $k+d+1$ corresponding paths for the request $\req{(f,l)}$, we retain all the necessary structure such that we do not create any new solutions in the reduced instance.

In the two following \texttt{for}-cycles, we first pick a vertex $y$ of $\Y$. This fixes the set of ancestors $\anc(y)$, i.e., it fixes the set of vertices on the path from $y$ to the root of its tree in $\F$. And, for this particular $y$, we then pick a sub-request $\subreq{(g,j)}$ which lives on this fixed set $\anc(y)$ and $M$.
This allows us to look only at some components of $G \setminus (M \cup \Y)$ and actually makes it possible for us to bound their number. The bounding happens on lines \ref{alg:mark_c_bound_start}--\ref{alg:mark_c_bound_end} and we can also say that the sub-request $\subreq{(g,j)}$ is resolved when the number of components is at least $2d$. The bound $2d$ follows from \autoref{lemma:dpvc_sol_mod}, which will be stated later.
For a resolved sub-request we proceed similarly to resolved request.
Namely, we preserve one corresponding path in each of some $2d$ of the components.

If the number of components is not large, the lines \ref{alg:mark_mark2_start}--\ref{alg:mark_mark2_end} run the second marking procedure \textbf{Mark~2} on each of these components and the aim is to bound their size.

Now, recall again, that in the intuition we examined some $d$-path $P$ in $G$ and some prospective solution $S$. The purpose of the marking procedure \textbf{Mark~2} is to brute-force all the possible ways of how the path $P$ and solution $S$ may compromise the paths which satisfy the sub-request $\subreq{(g,j)}$ and which are contained in the component $C_i$. The procedure works recursively, starts with the empty set of ``compromising'' vertices $W$ and it always picks a path which was not yet compromised, marks it (so that it remains in $\widehat{G}$), and tries to compromise its vertices one by one. By doing it like this, we ensure, that all the important parts of $C_i$ remain in $\widehat{G}$ no matter what parts of $C_i$ were compromised.

And the main trick is that we can stop the recursion of \textbf{Mark~2} once the number of ``compromising'' vertices reaches $2d$. This number $2d$ again follows from \autoref{lemma:dpvc_sol_mod}.

Now, the kernelization can be formally summarized as follows. Run the marking procedure \textbf{Mark} on the instance $(G,k)$. The marking results in two subsets $\widehat{V} \subseteq V(G)$ and $\widehat{E} \subseteq E(G)$ corresponding to marked vertices and edges by \textbf{Mark}.
Reduce the instance $(G,k)$ to the instance $(\widehat{G}, k)$ where $\widehat{G} = (\widehat{V}, \widehat{E})$.

With that we conclude the intuition and we continue with the formal proof of correctness. First, we state the crucial \autoref{lemma:dpvc_sol_mod}, then the main body of the proof follows in \autoref{lemma:dpvc_marking_correct}, and we finish with the proof of the size of the kernel in \autoref{lemma:dpvc_marking_size}.

\autoref{lemma:dpvc_sol_mod} roughly states that any reasonable solution only contains at most $d$ vertices among each set of components considered on line \ref{alg:components} of \autoref{alg:mark}.

%The marking procedure \textbf{Mark} consists of the following steps.
%\begin{enumerate}
%\item Mark all the vertices and edges in $G[M \cup \Y]$.
%\item For each resolved request $(f,l)$ pick arbitrary $k+d+1$ leaves $v_1, v_2, \dots, v_{k+d+1}$ of $\F[Y_{f,l}]$ and for each leaf $v_i$ pick an arbitrary path $P$ from $\mathcal{P}^{V(T_{v_i})}_{f,l,M}$ and mark the vertices and edges of $P$.
%\item For each vertex $y \in \Y$ and each sub-request $\widehat{(f,l)}$ such that $f \subseteq M \cup \anc(y)$, let $T^\C_{i_1}, T^\C_{i_2}, \dots, T^\C_{i_{q'}}$ be the subtrees such that $N(V(T^\C_{i_j})) \cap \Y \subseteq \anc(y)$ and $V(T^\C_{i_j})$ satisfies $\widehat{(f,l)}$.
%\begin{enumerate}
%\item If $q' \geq 2d+1$, for each $j \in \{1,2,\dots,2d+1\}$ pick an arbitrary path $P \in \PP^{V(T^\C_{i_j})}_{f,l,M \cup \Y}$ and mark the edges and vertices of $P$.
%\item Otherwise, for each $j \in \{1,2,\dots,q'\}$ run the marking procedure \textbf{Mark~2} on $\widehat{(f,l)}$ and $T^\C_{i_j}$.
%\end{enumerate}

%\end{enumerate}

\toappendix{\subsection{Proof of \autoref{lemma:dpvc_marking_correct} (Correctness of the Algorithm)}}%

%For the correctness of the kernelization algorithm we need the following crucial lemma.
%It shows that a solution that contains more than $d$ vertices ``under the same $y \in \Y$'' is not optimal.
%We use this crucially when dealing with sub-requests.

\begin{lemma}\label{lemma:dpvc_sol_mod}
Let $(G,k)$ be an instance of \textsc{$d$-PVC}. Let $\widehat{G} = (\widehat{V}, \widehat{E})$ be a subgraph of $G$ such that $\widehat{V}$ and $\widehat{E}$ are the result of running the marking procedure \textbf{Mark} on $(G,k)$. Let $S'$ be a solution for the instance $(\widehat{G}, k)$ of \textsc{$d$-PVC}. Let $y \in \Y$ and let $C_1, C_2, \dots, C_c$ be the vertex sets of the connected components of $\widehat{G} \setminus (M \cup \Y)$ such that $N(C_i) \cap \Y \subseteq \anc(y)$ for $i \in [c]$. Then $\widehat{S} = (S' \setminus \bigcup_{i \in [c]} C_i) \cup \anc(y)$ is a solution for $\widehat{G}$.
\end{lemma}

%\toappendix{%
%\begin{proof}[of \autoref{lemma:dpvc_sol_mod}]
\begin{proof}
If $\widehat{S}$ is a solution for $(\widehat{G},k)$, we are done.
Suppose to the contrary that it is not. Then there is a $d$-path $P$ in $\widehat{G} \setminus \widehat{S}$.
Assume that $P$ is selected such that it contains the least number of vertices which are in $S'$ i.e.,
$|V(P) \cap S'|$ is minimized among all $d$-paths in $\widehat{G} \setminus \widehat{S}$.
As $S' \setminus \widehat{S} \subseteq \bigcup_{i \in [c]}C_i$, path $P$ must contain at least one vertex from at least one set $C_i \cap S'$,
because otherwise $P$ would also be in $\widehat{G} \setminus S'$,
which is a contradiction with $S'$ being a solution for  $(\widehat{G},k)$.
Further, since $N(C_i) \subseteq (M \cup \Y)$, $N(C_i) \cap \Y \subseteq \anc(y)$, and $\anc(y) \subseteq \widehat{S}$ by assumption,
we have $N(C_i) \setminus \widehat{S} \subseteq M$. Therefore $P \cap M \neq \emptyset$, as otherwise the path $P$ would be contained in $C_i$,
which is a contradiction with $\M$ being a maximal packing of $d$-paths.

We split the path $P$ into segments according to the vertices of $M$, i.e., a \emph{segment} of $P$ is a sub-path $(v_1, v_2, \dots, v_s)$ of $P$ such that $v_2, v_3, \dots, v_{s-1} \in V(G) \setminus M$ and either $\{v_1, v_s\} \subseteq M$ (an inner segment), or
one of $v_1,v_s$ is in $M$, while the other is an endpoint of $P$ (an outer segment).
The argument is the same in both cases.
% without loss of generality, $v_1 \in M$ and $v_s \in V(G) \setminus M$ (an outer segment). We note that the distinction between inner and outer segments is not actually important for the proof itself, as the arguments work with both, we mainly mention them to bring both possibilities to the readers attention.

Let $P' = (v_1, v_2, \dots, v_s)$ be the segment of $P$ which uses some vertex from $C_i \cap S'$.
Observe, that the segment $P'$ corresponds to request $\req{(f,l)} = \req{(V(P') \cap M, s-1)}$ as $2 \leq s \leq d$ and $|V(P') \cap M| \in \{1,2\}$.
In particular, $V(P') \setminus M$ satisfies $\req{(f,l)}$.

We also know that $V(P') \cap \Y = \emptyset$, because $N(C_i) \setminus \widehat{S} \subseteq M$. With that we argue that the request $\req{(f,l)}$ must be resolved.
Indeed, suppose it is not.
By \autoref{observation:trees}(\ref{observation:connected_subgraph}) there is a vertex $v$ in $V(P') \cap C_i$ such that
that $V(P') \cap C_i \subseteq sub(v)$. But that implies that $\sub(v)$ satisfies the request $\req{(f,l)}$ and, therefore, vertex $v$ should have been included in $Y_{f,l}$ and, consequently, $v$ should have been included in $\Y$, which is a contradiction with $V(P') \cap \Y = \emptyset$.

Now, as the request $\req{(f,l)}$ is resolved, the marking procedure \textbf{Mark} picked $k+d+1$ leaves $h_1, h_2, \dots, h_{k+d+1}$ from $\F[Y_{f,l}]$ and for each such leaf $h_i$ it marked the vertices and edges of some path $P_i \in \PP^{\sub(h_i)}_{f,l}$. Therefore, these paths $P_1, P_2, \dots, P_{k+d+1}$ remained in $\widehat{G}$. Further, at least one of these paths is untouched by the vertices of $S'$ and the vertices of $P$ as $|S'| \leq k$ and $|V(P)| \leq d$, respectively. Let this one untouched path be $P_i$. Observe, that we can swap the segment $P'$ with the path $P_i$ in $P$ to obtain a $d$-path $P^*$. But then the path $P^*$ contains strictly fewer vertices which are in $S'$ than $P$, which is a contradiction with the choice of $P$.
\end{proof}
%}%to appendix

Now we can prove the correctness of the algorithm.

\begin{lemma}[\appmark]\label{lemma:dpvc_marking_correct}
Let $(G,k)$ be an instance of \textsc{$d$-PVC}. Let $\widehat{G} = (\widehat{V}, \widehat{E})$ be a subgraph of $G$ such that $\widehat{V}$ and $\widehat{E}$ are obtained by running the marking procedure \textbf{Mark} on $(G,k)$. Then, $(G,k)$ is a YES instance if and only if $(\widehat{G}, k)$ is a YES instance.
\end{lemma}
\begin{proof}[Proof sketch.]
For the ``if'' direction we pick a solution $S'$ to $\widehat{G}$ such that any application of \autoref{lemma:dpvc_sol_mod} would increase its size.
If $S'$ was not a solution to $G$, then as in \autoref{lemma:dpvc_sol_mod}, we pick a special $d$-path $P$ witnessing that, but this time with the least number of unmarked edges in $G$.
Then we again split $P$ into segments according to $M$ and, in case the corresponding request was not resolved, further into sub-segments according to $\Y$. We always pick a \mbox{(sub-)segment} with at least one unmarked edge and we show that we can swap the (sub-)segment with some other suitable fully marked sub-path to obtain a contradiction with the choice of~$P$.
% or we will show that there actually cannot be an unmarked  edge in the chosen (sub-)segment.
\end{proof}

\toappendix{%
% The rest of this subsection is devoted to the proof of \autoref{lemma:dpvc_marking_correct}.\\
%
% \begin{proof}
As $\widehat{G}$ is a subgraph of $G$, any solution for $(G,k)$ is also a solution for $(\widehat{G},k)$, proving one of the implications.
Therefore, we will focus mainly on the other implication.

Let $S'$ be a solution for $(\widehat{G},k)$. We start by constructing another solution $\widehat{S}$ by repeatedly applying \autoref{lemma:dpvc_sol_mod} on $S'$ whenever it does not increase the size of the solution. Formally, we construct a sequence $S'_0, S'_1, \dots, S'_{s'}$ where $S'_0 = S'$, $\widehat{S} = S'_{s'}$, $S'_i$ is the result of applying \autoref{lemma:dpvc_sol_mod} on $S'_{i-1}$ in such a way that $|S'_i| \leq |S'_{i-1}|$ for all $i \in [s']$, and it holds for all $y \in \Y$ that applying \autoref{lemma:dpvc_sol_mod} on $S'_{s'}$ with $y$ results in a solution larger than $S'_{s'}$.

If $\widehat{S}$ is a solution for $(G,k)$, we are done. Assume on contrary that it is not. Therefore, there must be a $d$-path $P$ in $G \setminus \widehat{S}$. Assume that $P$ is selected such that $P$ has the least number of unmarked edges from all the $d$-paths in $G \setminus \widehat{S}$.
As $\widehat{S}$ is a solution for $\widehat{G}$, at least one edge must be unmarked in $P$, otherwise, $P$ would also be a $d$-path in $\widehat{G} \setminus \widehat{S}$.

Observe, that $V(P) \cap M \neq \emptyset$, otherwise, we would have a contradiction with $\M$ being a maximal packing of $d$-paths.

We again split the path $P$ into segments according to the vertices of $M$, i.e., a \emph{segment} of $P$ is a sub-path $(v_1, v_2, \dots, v_s)$ of $P$ such that $v_2, v_3, \dots, v_{s-1} \in V(G) \setminus M$ and either $\{v_1, v_s\} \subseteq M$ (an inner segment), or
one of $v_1,v_s$ is in $M$, while the other is an endpoint of $P$ (an outer segment).
The argument is again the same in both cases.
% without loss of generality, $v_1 \in M$ and $v_s \in V(G) \setminus M$ (an outer segment). We note that the distinction between inner and outer segments is not actually important for the proof itself, as the arguments work with both, we mainly mention them to bring both possibilities to the readers attention.

Pick a segment $P'=(v_1, v_2, \dots, v_s)$ of $P$, such that there is at least one edge unmarked in $P'$.
Such a segment must exist, because each (unmarked) edge of $P$ belongs to some segment.
Since all the edges in $G[M]$ were marked, we have $V(P') \not\subseteq M$, i.e., it has at least one vertex outside of $G[M]$.
Note that the length $s-1$ of the segment $P'$ is at least $1$ and at most $d-1$ as it contains vertices from $M$, $V(P') \not\subseteq M$, and $P$ is a $d$-path.

%For now, assume that we are dealing with an inner segment, i.e., $v_1, v_s \in M$, as we will point out the differences in the argumentation for an outer segment later.

Observe, that the segment $P'$ corresponds to request $\req{(f,l)} = \req{(V(P') \cap M, s-1)}$ as $1 \leq s-1 \leq d-1$ and $|V(P') \cap M| \in \{1,2\}$.
In particular, $V(P) \setminus M$ satisfies $\req{(f,l)}$.

Now, suppose that request $\req{(f,l)}$ is resolved. The marking procedure \textbf{Mark} picked for such request $k+d+1$ leaves $h_1, h_2, \dots, h_{k+d+1}$ of $\F[Y_{f,l}]$ and for each such leaf $h_i$ it marked the edges and vertices of some path $P_i \in \PP^{\sub(h_i)}_{f,l}$.
Note, that these marked paths have length $s-1$, have vertices of $f$ as their endpoints, and are pairwise disjoint except for the vertices of $f$. Indeed, the paths $P_i$ and $P_j, i \neq j$ were picked from $\PP^{\sub(h_i)}_{f,l}$ and $\PP^{\sub(h_j)}_{f,l}$ and the sets $\sub(h_i)$ and $\sub(h_j)$ are disjoint by the \autoref{observation:trees}(\ref{observation:trees_disjoint}) as the leaves $h_i$ and $h_j$ are not in ancestor-descendant relation. Moreover, at least one of the paths $P_1, P_2, \dots, P_{k+d+1}$ is touched neither by the solution $\widehat{S}$, nor by the path $P$ as $|\widehat{S}| \leq k$ and the length of $P$ is $d-1$. Let this untouched path be $P_i$.
Observe, that we can swap $P_i$ with the segment $P'$ in $P$ to obtain a $d$-path $P^*$.
Therefore, $P^*$ is also a $d$-path in $G \setminus \widehat{S}$ but,  since  all the edges of $P_i$ were marked, it has less unmarked edges than $P$, which is a contradiction with the choice of $P$.

Next, we assume that request $\req{(f,l)}$ is not resolved. Since $P' \setminus M$ is a connected subgraph of $G'$, by \autoref{observation:trees}(\ref{observation:connected_subgraph}) there is a vertex $w \in V(P') \setminus M$ such that $V(P') \setminus M \subseteq \sub(w)$. Then, by definition, $w \in Y_{f,l}$ and  $w \in \Y$, i.e., $V(P') \cap \Y \neq \emptyset$.

We zoom in on the segment $P'$ and split it further into sub-segments according to vertices of $\Y$.
A \emph{sub-segment} of $P'$ is a sub-path $(v'_1, v'_2, \dots, v'_{s'})$ of $P'$ such that $v'_2, v'_3, \dots, v'_{s'-1} \in V(G) \setminus (M \cup \Y)$ and either $\{v'_1, v'_{s'}\} \in (M \cup \Y)$ (an inner sub-segment), or one of $v'_1,v'_{s'}$ is in $M \cup \Y$, while the other is an endpoint of $P$ (and $P'$) (an outer sub-segment).
Yet again, the argument is the same in both cases.
%
% Again, we note that the distinction between inner and outer sub-segments is not actually important for the proof itself, we just want to bring both possibilities to the readers attention.
%

Using similar strategy to the one we used to pick a segment $P'$, we pick a sub-segment $P''=(v'_1, v'_2, \dots, v'_{s'})$ of $P'$, such that there is at least one edge unmarked in $P''$.
Such a sub-segment must exist, because each (unmarked) edge of $P'$ belongs to some sub-segment.
Since all edges in $G[M \cup \Y]$ were marked, $V(P'') \not\subseteq (M \cup \Y)$, i.e., it has at least one vertex outside of $G[M \cup \Y]$.

We provide the following observations about the sub-segment $P''$.

\begin{observation}\label{observation:subsegments}
\begin{enumerate}[(a)]
\item The length $s'-1$ of $P''$ is at least $1$ and at most $d-1$.
\item At least one endpoint of $P''$ is in $\Y$, i.e., $|V(P'') \cap \Y| \ge 1$.
\item There is some connected component of $G \setminus (M \cup \Y)$ with vertex set $C_i$ such that $V(P'') \setminus (M \cup \Y) \subseteq C_i$. \label{observation:subsegments_comp}
\item If $\{v'_1, v'_{s'}\} \subseteq \Y$ then the vertices $v'_1, v'_{s'}$ are in an ancestor-descendant relation in $\F[\Y]$.
%\item No sub-segment $(v'_1, v'_2, \dots, v'_{s'})$ of $P'$ has $\{v'_1, v'_{s'}\} \subseteq M$.
%\item For each sub-segment $(v'_1, v'_2, \dots, v'_{s'})$ of $P'$ there is a subtree $T^\C_i$ such that the inner vertices $\{v'_2, v'_3, \dots, v'_{s'-1}\} \subseteq V(T^\C_i)$.
%\item For each sub-segment $(v'_1, v'_2, \dots, v'_{s'})$ of $P'$ such that $\{v'_1, v'_{s'}\} \subseteq \Y$ the vertices $v'_1, v'_{s'}$ are in an ancestor-descendant relation in $\F[\Y]$.
\end{enumerate}
\end{observation}
\begin{proof}
\begin{enumerate}[(a)]
\item The length of $P''$ is at least $1$, because it has at least one vertex in $(M \cup \Y)$ and one vertex outside of $(M \cup \Y)$. The length of $P''$ is at most $d-1$ as it is a sub-path of the segment $P'$ whose length is at most $d-1$.
\item We split $P'$ to sub-segments according to the vertices of $\Y$. If neither of the endpoints of $P''$ is in $\Y$, it means that $P''$ is the whole $P'$ contradicting that $V(P') \cap \Y \neq \emptyset$.
\item The path $P'' \setminus (M \cup \Y)$ is a connected subgraph of $G \setminus (M \cup \Y)$. Hence it must be contained within one connected component.
\item Let $C_i$ be as in (\ref{observation:subsegments_comp}). Then $\{v'_1, v'_{s'}\} \subseteq N(C_i)$ and any pair of vertices in $N(C_i) \cap \Y$ is in an ancestor-descendant relation. \qedhere
\end{enumerate}
\end{proof}

%$$$$$$$$$$$$$$$$$$$$$$$$$$$$$$$$$$$$$$$$$$$$$$$$$

Let $C_i$ be the vertex set of the connected component of $G \setminus (M \cup \Y)$ from the \autoref{observation:subsegments}(\ref{observation:subsegments_comp}), i.e., such that $V(P'') \setminus (M \cup \Y) \subseteq C_i$.
By \autoref{observation:trees}(\ref{observation:connected_subgraph}) there is a vertex $w \in C_i$ such that $C_i \subseteq \sub(w)$.
Then by \autoref{observation:trees}(\ref{observation:trees_disjoint}), we have $N(C_i) \cap \Y \subseteq \anc(w)$.
Let $y$ be the vertex in $N(C_i) \cap \Y$ farthest from the root of the appropriate tree of $\F$ %with the highest depth in $\F$
(closest to the leaves).
Then also $N(C_i) \cap \Y \subseteq \anc(y)$.

As a corollary of \autoref{observation:subsegments}, the sub-segment $P'' = (v'_1, v'_2, \dots, v'_{s'})$ corresponds to a sub-request $\subreq{(g,j)} = \subreq{(V(P'') \cap (M \cup \Y),s'-1)}$.
Note that $V(P'') \setminus (M \cup \Y)$ satisfies $\subreq{(g,j)}$.

Now, suppose that the marking procedure \textbf{Mark} for the vertex $y$ and the sub-request $\subreq{(g,j)}$ found at least $2d$ vertex sets $C'_1, C'_2, \dots, C'_{2d}$ of connected components of $G \setminus (M \cup \Y)$ such that for all $i \in [2d]$, $N(C'_i) \cap \Y \subseteq \anc(y)$ and $C'_i$ satisfies $\subreq{(g,j)}$.

Then, at most $d-1$ vertices of the solution $\widehat{S}$ are contained in $C'_1, C'_2, \dots, C'_{2d}$, i.e., $|\widehat{S} \cap \bigcup_{i \in [2d]} C_i| \leq d-1$, as otherwise \autoref{lemma:dpvc_sol_mod} should have been applied to $\anc(y)$ and the vertex sets $C'_1, C'_2, \dots, C'_{2d}$ as $|\anc(y)| \leq d-1$ due to \autoref{observation:trees}(\ref{observation:trees_depth}).

Further, only $d$ of the vertex sets $C'_1, C'_2, \dots, C'_{2d}$ could be touched by the path~$P$.
Therefore, there is at least one vertex set $C'_i$ which is touched by neither the solution $\widehat{S}$, nor the path $P$.
The marking procedure \textbf{Mark} picked an arbitrary path $P_i$ from $\PP^{C'_i}_{g,j}$ and marked its vertices and edges. Observe, that we can swap $P_i$ with the sub-segment $P''$ in $P$ to obtain a $d$-path $P^*$. Therefore, $P^*$ is also a $d$-path in $G \setminus \widehat{S}$ but, since all the edges of $P_i$ were marked, it has less unmarked edges than $P$, which is a contradiction with the choice of $P$.

Finally, suppose that at most $q' < 2d$ vertex sets $C'_1, C'_2, \dots, C'_{q'}$ were found by the marking procedure \textbf{Mark} for the vertex $y$ and the sub-request $\subreq{(g,j)}$. This means that the vertex set $C_i$ must be among the vertex sets $C'_1, C'_2, \dots, C'_{q'}$, i.e., $C_i = C'_z$ for some $z \in [q']$, as $C_i$ also satisfies the sub-request $\subreq{(g,j)}$ and $N(C_i) \cap \Y \subseteq \anc(y)$. Therefore, we know that the marking procedure \textbf{Mark 2} was called with $\Big(\subreq{(g,j)}, C_i, \emptyset \Big)$.

Let $\widehat{W} = \Big(\widehat{S} \cup \big(V(P) \setminus V(P'')\big)\Big) \cap C_i$. Then, the size of $\widehat{W}$ is at most $2d-1$. Indeed, first, we have $|\widehat{S} \cap C_i| < d$ as otherwise the \autoref{lemma:dpvc_sol_mod} should have been applied with $\anc(y)$ and the vertex set $C_i$, and, second, $|(V(P) \setminus V(P'')) \cap C_i| < d$ as $P$ is a $d$-path. We are ready to state the following observation.

\begin{observation}\label{observation:mark2marks}
If the marking procedure \textbf{Mark 2} is called with $\Big(\subreq{(g,j)}, C_i, W \Big)$ such that $W \subseteq \widehat{W}$, then
\begin{enumerate}[(i)]
\item either \textbf{Mark 2} marks the vertices and edges of some path $\widehat{P} \in \PP^{C_i \setminus \widehat{W}}_{g,j}$,
\item or \textbf{Mark 2} makes a recursive call with $\Big(\subreq{(g,j)}, C_i, W' \Big)$ such that $W' \subseteq \widehat{W}$ and $|W'| > |W|$.\label{obs:marks2marks:item2}
\end{enumerate}

\end{observation}
\begin{proof}
First, note that the set $\PP^{C_i \setminus \widehat{W}}_{g,j} \neq \emptyset$ as the sub-segment $P''$ is the proof of that. Thus, the set $\PP^{C_i \setminus W}_{g,j} \neq \emptyset$ as $W \subseteq \widehat{W}$. Further, $|W| \leq |\widehat{W}| < 2d$. Consequently, the condition on the first line of \textbf{Mark~2} is satisfied.
Therefore, \textbf{Mark~2} picks some path $P \in \PP^{C_i \setminus W}_{g,j} \neq \emptyset$.
If $P \in \PP^{C_i \setminus \widehat{W}}_{g,j}$, we are done. Suppose that $P \not\in \PP^{C_i \setminus \widehat{W}}_{g,j} \neq \emptyset$. Then, $V(P) \cap (\widehat{W} \setminus W) \neq \emptyset$ and, as the marking procedure \textbf{Mark~2} proceeds by making recursive calls for each vertex in $V(P)$, it will in particular make a recursive call with some vertex $v \in V(P) \cap (\widehat{W} \setminus W)$. This recursive call will be made with the set $W' = W \cup \{v\}$ which, therefore, satisfies $W' \subseteq \widehat{W}$ and $|W'| > |W|$.
\end{proof}

We know that a call to \textbf{Mark 2} was made with $\Big(\subreq{(g,j)}, C_i, \emptyset \Big)$.
Consider the largest set $\widetilde{W} \subseteq \widehat{W}$ such that a call $\Big(\subreq{(g,j)}, C_i, \widetilde{W} \Big)$ was made.
Such a set exists since $|\widehat{W}| < 2d$ and $\emptyset \subseteq \widehat{W}$.
Since \autoref{observation:mark2marks} applies to this call and (\ref{obs:marks2marks:item2}) cannot hold by the choice of $\widetilde{W}$, we know that the vertices and edges of some path $\widehat{P} \in \PP^{C_i \setminus \widehat{W}}_{g,j}$ were marked. But, either the path $\widehat{P}$ is exactly the same path as $P''$, which would be a contradiction with the fact that there is an unmarked edge in $P''$, or $\widehat{P}$ can again be exchanged with $P''$ in $P$ to obtain a $d$-path $P^*$. Then $P^*$ is a $d$-path in $G \setminus \widehat{S}$ with less unmarked edges than $P$, which is a contradiction with the choice of $P$.

%proof summary
We conclude the proof with a brief summary of the proof.

We started by picking a $d$-path $P$ with the least number of unmarked edges. We have shown that there must be a segment $P'$ of $P$ with an unmarked edge, this segment $P'$ either corresponds to a resolved request, or to an unresolved request.

In the case of a resolved request, we have shown that the segment $P'$ can be replaced by some other path with no unmarked edge, resulting in a contradiction with the choice of $P$.

In the case of an unresolved request, we continued by further splitting the segment $P'$ into sub-segments, where one of the sub-segments contained an unmarked edge.

This sub-segment $P''$, again, corresponded to some sub-request for which we have shown that either enough paths were marked for this sub-request, which allowed us, again, to replace the sub-segment $P''$ with some other path with no unmarked edge, resulting in a contradiction with the choice of $P$, or that the marking procedure \textbf{Mark 2} was called with suitable parameters.

This call resulted in either a contradiction with sub-segment $P''$ containing an unmarked edge, or in a possibility to, again, replace the sub-segment $P''$ with some other path with no unmarked edge, resulting in a contradiction with the choice of $P$.
% \end{proof}
}% to appendix

\toappendix{\subsection{Proof of \autoref{lemma:dpvc_marking_size} (Size of the Kernel)}}%

The following lemma shows the bound on the size of the kernel.

\begin{lemma}[\appmark]\label{lemma:dpvc_marking_size}
Let $(G,k)$ be an instance of \textsc{$d$-PVC}. Let $\widehat{G} = (\widehat{V}, \widehat{E})$ be a subgraph of $G$ such that $\widehat{V}$ and $\widehat{E}$ are obtained by running the marking procedure \textbf{Mark} on $(G,k)$. Then, $|\widehat{V}| = \bigo(k^4d^{2d+9})$ and $|\widehat{E}| = \bigo(k^4d^{2d+9})$.
\end{lemma}
% \begin{proof}
\toappendix{%
Let us follow the lines of the marking procedure \textbf{Mark}. First, it marks all the vertices and edges in $G[M \cup \Y]$.

% We will do this by counting the number of vertices and edges in $G[M]$ and $G[\Y]$ and counting the number of edges between $M$ and $\Y$ in $G$.
%
% \begin{claim}\label{claim:size_of_M}
% There are at most $dk$ vertices and $\bigo(d^2k^2)$ edges in $G[M]$.
% \end{claim}
% \begin{claimproof}
% As $|M| \leq dk$, the number of edges in $G[M]$ is at most $\binom{dk}{2} = \bigo(d^2k^2)$.
% \end{claimproof}

To count the number of vertices and edges in $G[M \cup \Y]$, we first count the maximum number of possible requests, then derive the maximum number of leaves of the subforest $F[\Y]$, which will ultimately lead us to the number of vertices and edges in $G[M \cup \Y]$.\\

\begin{claim}\label{claim:number_of_requests}
There are at most $(\binom{dk}{2} + dk)(d-1) = \bigo(k^2d^3)$ possible requests.
\end{claim}
\begin{claimproof}
For each possible request $\req{(f,l)}$ we have that $f \subseteq M$ and $|f| \in \{1,2\}$. As $|M| \leq dk$, the number of options to pick the set $f$ is $\binom{dk}{2}$ when $|f| = 2$ and $dk$ when $|f| = 1$. For each unique set $f$ we have $d-1$ different choices for~$l$. Therefore the total number of possible request is at most $(\binom{dk}{2} + dk)(d-1)$ which is $\bigo(k^2d^3)$.
\end{claimproof}

\begin{claim}
The subforest $\F[\Y]$ has at most $\bigo(k^3d^4)$ leaves.
\end{claim}
\begin{claimproof}
By the previous claim, there are at most $\bigo(k^2d^3)$ possible requests and for each such request we add at most $k+d+1$ leaves to $F[\Y]$ as we include only those sets $Y_{f,l}$ corresponding to unresolved requests which, by definition, have at most $k+d+1$ leaves in $\F[Y_{f,l}]$. Therefore, the maximum number of leaves in $\F[\Y]$ is $\bigo((k^2d^3)\cdot(k+d+1)) = \bigo(k^3d^4)$.
\end{claimproof}

\begin{claim}
There are at most $\bigo(k^3d^5)$ vertices and $\bigo(k^3d^6)$ edges in $G[\Y]$.
\end{claim}
\begin{claimproof}
The forest $\F$ consists of depth-first search trees. By the previous claim, there are at most $\bigo(k^3d^4)$ leaves in $\F[\Y]$, each of them having at most $d-1$ ancestors, which together gives us, that there is at most $\bigo(k^3d^4) \cdot (d-1) = \bigo(k^3d^5)$ vertices in $\F[\Y]$ and, in particular, in $\Y$.

Now, let us recall the key property of depth-first search trees. Given a depth-first search tree or, in our case, forest $\F$, one can partition the edges of $G$ into two types according to $\F$. There are tree edges -- the edges in the forest $\F$, and back edges -- the edges of $G$ which are not in $\F$ and a vertex of $\F$ can only have back edges to its ancestors. Therefore, in order to count the number of edges in $G[\Y]$, it suffices to count the number of edges in $\F[\Y]$ and the number of possible back edges for each vertex in $\F[\Y]$.

The number of edges in $\F$ is at most $\bigo(k^3d^5)$ as there are at most $\bigo(k^3d^5)$ vertices in $\F$. Each vertex $y \in \Y$ can have at most $d-3$ back edges as it has at most $d-1$ ancestors including itself. We get that there are at most $\bigo(k^3d^5)\cdot(d-2) = \bigo(k^3d^6)$ possible back edges for vertices in $\F[\Y]$. The total number of edges in $G[\Y]$ is, therefore, $\bigo(k^3d^5) + \bigo(k^3d^6) = \bigo(k^3d^6)$.
\end{claimproof}

\begin{claim}
There are at most $\bigo(k^4d^6)$ edges between $M$ and $\Y$ in $G$.
\end{claim}
\begin{claimproof}
By the previous claims, there are at most $dk$ vertices in $M$ and at most $\bigo(k^3d^5)$ in $\Y$. Therefore, there are at most $dk \cdot \bigo(k^3d^5) = \bigo(k^4d^6)$ possible edges in between.
\end{claimproof}

Now, it suffices to sum the number of vertices and edges in $G[M], G[\Y]$ and in between $M$ and $\Y$ in $G$ and we obtain that there are at most $\bigo(k^3d^5)$ vertices and $\bigo(k^4d^6)$ edges in $G[M \cup \Y]$ marked on line \ref{alg:mark_my}.

The marking procedure \textbf{Mark} continues with marking $k+d+1$ paths for each resolved request. Each such path has length at most $d-1$. By the previous claims, there are at most $\bigo(k^2d^3)$ requests, therefore the procedure on lines \ref{alg:mark_foreach_resolved_start}--\ref{alg:mark_foreach_resolved_end} marks at most $\bigo(k^2d^3) \cdot (k+d+1) \cdot d = \bigo(k^3d^5)$ vertices and edges.

Finally, the marking procedure marks some vertices and edges for each $y \in \Y$ and each sub-request $\subreq{(g,j)}$ such that $g \subseteq M \cup \anc(y)$ and $g \not\subseteq M$.

\begin{claim}
For each $y \in \Y$, there are at most $\bigo(kd^3)$ sub-requests $\subreq{(g,j)}$ such that $g \subseteq M \cup \anc(y)$ and $g \not\subseteq M$.
\end{claim}
\begin{claimproof}
Let $y\in \Y$. We have $\anc(y) \leq d-1$. The set $g$ can attain one of three forms. First, $|g| = 2$, $|g \cap M| = 1$ and $|g \cap \anc(y)| = 1$, second, $|g| = 2$ and $g \subseteq \anc(y)$, and third, $|g| = 1$ and $g \subseteq \anc(y)$. The maximum number of options to pick the set $g$ is therefore $dk \cdot (d-1) = \bigo(kd^2)$, $\bigo(d^2)$, and $\bigo(d)$ in these cases, respectively. Hence, there are at most $\bigo(kd^2)$ options to pick the set $g$ and for each such set $g$, we have $d-1$ choices for $l$. Therefore, there are at most $\bigo(kd^3)$ sub-requests $\subreq{(g,j)}$ for the vertex $y$.
\end{claimproof}

\begin{claim}
For each $y \in \Y$ and each sub-request $\subreq{(g,j)}$ such that $g \subseteq M \cup \anc(y)$ and $g \not\subseteq M$, the marking procedure \textbf{Mark} marks at most $\bigo(d^{2d+1})$ vertices and edges.
\end{claim}
\begin{claimproof}
Let $C_1,C_2,\dots,C_q'$ be the vertex sets as described in the procedure \textbf{Mark}. If $q' \geq 2d$, the procedure \textbf{Mark} picks $2d$ paths and marks their vertices and edges. As the paths have length at most $d-1$, in this case, the procedure marks at most $2d \cdot (d-1) = \bigo(d^2)$ vertices and edges.

If $q' < 2d$, the procedure \textbf{Mark} for each component $C_i$ calls the marking procedure \textbf{Mark~2}. The procedure \textbf{Mark~2} is a recursive procedure, therefore, we proceed by analyzing its recursion tree.

The stopping condition on the first line ensures that the recursion tree has depth at most $2d$ as each recursive call increases the size of the set $W$. In each recursive call, the procedure branches into at most $d$ recursive calls as $|V(P)| \leq d$ for each $P \in \PP^{C_i \setminus W}_{g,j}$. Therefore, the size of the recursion tree is at most $2d^{2d} = \bigo(d^{2d})$. Further, in each recursive call, the procedure \textbf{Mark~2} marks at most $d$ edges and vertices. Therefore, the procedure \textbf{Mark~2} marks at most $\bigo(d^{2d}) \cdot d = \bigo(d^{2d+1})$ vertices and edges.
\end{claimproof}

Putting the previous claims together, we have that $|\Y| \leq \bigo(k^3d^5)$, for each vertex in $\Y$ there are at most $\bigo(kd^3)$ sub-requests, and for each vertex in $\Y$ and each sub-request for $y$, the procedure marks at most $\bigo(d^{2d+1})$ edges and vertices. Therefore, the total number of marked vertices and edges on lines \ref{alg:mark_foreach_y_start}--\ref{alg:mark_foreach_y_end} is $\bigo(k^3d^5) \cdot \bigo(kd^3) \cdot \bigo(d^{2d+1}) = \bigo(k^4d^{2d+9})$.

Summing over all lines of the marking procedure \textbf{Mark}, we obtain that the total number of marked vertices is $\bigo(k^3d^5) + \bigo(k^3d^5) + \bigo(k^4d^{2d+9}) = \bigo(k^4d^{2d+9})$ and the total number of marked edges is $\bigo(k^4d^6)+ \bigo(k^3d^5) + \bigo(k^4d^{2d+9}) = \bigo(k^4d^{2d+9})$.
% \end{proof}
} % to appendix

%\begin{observation}
%The subforest $\F[\Y]$ has at most $\bigo(k^3d^4)$ leaves.
%\end{observation}
%\begin{proof}
%
%\end{proof}

%\begin{observation}\label{observation:number_of_requests}
%There are at most $(\binom{dk}{2} + dk)(d-1) = \bigo(k^2d^3)$ possible requests.
%\end{observation}
%\begin{proof}
%
%\end{proof}

We summarize the result in the following theorem.

\begin{theorem}
\textsc{$d$-Path Vertex Cover} admits a kernel with $\bigo(k^4d^{2d+9})$ vertices and edges, where $k$ is the size of the solution.
\end{theorem}
\begin{proof}
As the marking procedures \textbf{Mark} and \textbf{Mark~2} can by implemented in polynomial time, the theorem directly follows from Lemmas \ref{lemma:dpvc_marking_correct} and \ref{lemma:dpvc_marking_size}.
\end{proof}

\section{Conclusion}

We presented kernels with $\bigo(k^2)$ edges for \textsc{4-PVC} and \textsc{5-PVC} and with $\bigo(k^4d^{2d+9})$ edges for \dPVC{} for any $d\geq6$.
An obvious open question is whether there is a kernel with $\bigo(k^2)$ edges for every $d\geq6$.

Furthermore, the size of our kernel depends on $d$ by a factor of $d^{\bigo(d)}$.
We believe that this could be improved to $2^{\bigo(d)}$ with the use of representative sets.
However, improving this to a factor polynomial in $d$ would imply coNP $\subseteq$ NP/poly.
As observed by Dell and Marx~\cite{DellM18}, running such a kernel with $k=0$ would give a polynomial kernel for the $d$-Path problem,
which would have the above mentioned implications.

Next, for \textsc{2-PVC} and \textsc{3-PVC}, there are kernels with linear number of vertices~\cite{Lampis11, XiaoK17}.
Hence, another open question is whether such a kernel can be obtained also for say \textsc{4-PVC}.
Further interesting open questions can be found in the recent survey of Tu~\cite{Tu22}.

%%
%% Bibliography
%%

%% Please use bibtex,

\bibliography{main}

\ifdefined\mainappendix
\clearpage
\appendix

\appendixText
\fi

\end{document}